\documentclass[journal]{IEEEtran}

\IEEEoverridecommandlockouts

\usepackage{caption}
\usepackage{subcaption}
\usepackage{graphicx}
\usepackage{cite}
\usepackage{url}
\usepackage[cmex10]{amsmath}
\usepackage{amssymb}
\usepackage{amsthm}
\usepackage{algorithm}
\usepackage{algorithmicx}
\usepackage{algpseudocode}
\usepackage{cases}
\usepackage{color}
\usepackage{soul}
\usepackage{xcolor}
\usepackage{framed}
\usepackage{cite}
\usepackage{epstopdf}
\usepackage{calc}
\usepackage{array}
\usepackage{amsmath}
\usepackage{comment}
\usepackage{bbm}
\colorlet{shadecolor}{yellow}
\DeclareGraphicsExtensions{.pdf,.eps,.png,.jpg,.mps}

\newtheorem{theorem}{Theorem}
\newtheorem{corollary}{Corollary}

\begin{document}
\title{QoS Aware Power Allocation and User Selection in Massive MIMO Underlay Cognitive Radio Networks}

%
\author{Shailesh Chaudhari \textit{Student Member, IEEE},  Danijela Cabric, \textit{Senior Member, IEEE}%
\thanks{Shailesh Chaudhari and Danijela Cabric are with the Department of Electrical Engineering, University of California, Los Angeles, 56-125B Engineering IV Building, Los Angeles, CA 90095-1594, USA (email: schaudhari@ucla.edu, danijela@ee.ucla.edu).}
%
%
\thanks{This work has been supported by the National Science Foundation under CNS grant 1149981.}
}


\maketitle

\begin{abstract}
We address the problem of power allocation and secondary user (SU) selection in the downlink from a secondary base station (SBS) equipped with a large number of antennas in an underlay cognitive radio network. A new optimization framework is proposed in order to select the maximum number of SUs and compute power allocations in order to satisfy instantaneous rate or QoS requirements of SUs. The optimization framework also aims to restrict the interference to primary users (PUs) below a predefined threshold using available imperfect CSI at the SBS. In order to obtain a feasible solution for power allocation and user selection, we propose a low-complexity algorithm called Delete-su-with-Maximum-Power-allocation (DMP). Theoretical analysis is provided to compute the interference to PUs and the number of SUs exceeding the required rate. 
The analysis and simulations show that the proposed DMP algorithm outperforms the state-of-the art selection algorithm in terms of serving more users with minimum rate constraints, and it approaches the optimal solution if the number of antennas is an order of magnitude greater than the number of users.

\end{abstract}

\IEEEpeerreviewmaketitle



\section{Introduction}
\label{sec:Introduction}
Due to increasing number of wireless devices and data rate demands, researchers are looking for various techniques to improve the spectrum efficiency of 5G wireless networks and serve a large number of devices in the available spectrum. Massive MIMO and underlay cognitive radio are  being considered for 5G networks in order to accommodate more devices in the available spectrum \cite{andrews2014,boccardi2014, gupta2015}. In a massive MIMO system, a base station equipped with a large number of antennas serves multiple users using beamforming techniques in the same time-frequency resource block  \cite{ngo2013a}. On the other hand, in an underlay cognitive radio (CR) network, a secondary base station (SBS) serves its users (secondary users) while keeping the interference to licensed primary users (PUs) below a specified threshold \cite{biglieri2012}. In underlay CR networks, the SBS transmits the downlink signal in the same time-frequency resource block as the primary transmitter. This is different from traditional interweave cognitive networks where the SBS transmits in an orthogonal time-frequency resource block.

The secondary BS, if equipped with a large number of antennas, can potentially employ beamforming techniques and serve multiple  secondary users (SUs) in the downlink while limiting the interference to primary receivers. However, due to imperfect knowledge of the channels between PUs and the SBS, interference constraints at the PUs, and different rate requirements of SUs, it may not be feasible to serve all the SUs in the network \cite{chaudhari2017a}. Therefore, a judicious selection of SUs and power allocation are required at the SBS in order to simultaneously serve multiple SUs with required rates while limiting the interference to PUs.

\subsection{Related Work}
Underlay CR networks with multiple antenna systems have received attention in recent years, since such networks allow concurrent transmissions from  primary and secondary transmitters, thereby increasing the spectrum efficiency \cite{yiu2012a, du2012b, du2013, tsinos2013a,noam2013, noam2014, wang2015, al-ali2016, xiong2016}. The works in \cite{yiu2012a, du2012b, tsinos2013a, noam2013, noam2014, al-ali2016} consider small-scale MIMO with approximately ten or fewer antennas at the secondary transmitter. These works consider only one SU in the system, therefore they do not require SU selection. User selection mechanism has been partially considered in \cite{du2013, xiong2016}. An indirect selection mechanism is implemented in \cite{du2013} where SUs receiving less than 0dB signal-to-interference-plus-noise-ratio (SINR) are dropped from the downlink transmission. The selection algorithm in \cite{xiong2016} needs to know the number of users to be selected. A massive MIMO system has been employed in the secondary systems in \cite{wang2015, chaudhari2015, chaudhari2017a}. In our previous works \cite{chaudhari2015, chaudhari2017a}, we proposed to use massive MIMO to serve multiple SUs concurrently with primary transmission, while the algorithm in \cite{wang2015} still serves only one SU. A selection algorithm under line-of-sight channels is proposed in \cite{chaudhari2015}, while the feasibility of serving all SUs under Rayleigh fading channels is studied in \cite{chaudhari2017a}. The selection algorithm was not considered in \cite{chaudhari2017a}.

Massive MIMO systems differ from small-scale MIMO systems in \cite{du2012b, du2013, tsinos2013a, al-ali2016} in the design of beamforming (or precoding) vectors. In small-scale MIMO systems, optimum beamforming vectors are computed using iterative algorithms \cite{du2012b, du2013, al-ali2016}, and interference alignment \cite{tsinos2013a}. Such approaches become prohibitively expensive in terms of complexity when used with massive MIMO systems. Using linear beamforming techniques such as zero-forcing (ZF), maximal-ratio combining (MRC) or minimum mean-square error (MMSE), the beamforming vectors can be computed using closed-form expressions without requiring any iterative search if the channels and the selected user set is known. In this paper, we focus on ZF beamforming as it is can also be used to restrict the interference toward PUs.


In an underlay CR network, the interference at primary receivers (PRs) can be eliminated using ZF beamforming if the channels between PRs and the SBS are perfectly known at the SBS. However, due to imperfect CSI in practical networks, there is non-zero leakage interference transmitted towards PRs even when ZF beamforming is used. The magnitude of the interference depends on the power allocated to SUs as well as the set of SUs selected. Therefore, there is a need to design a robust interference control mechanism along with power allocation and user selection in order to limit the interference to PRs below a specified threshold.

\subsection{Summary of Contributions and Outline}
The main contributions of this work are summarized below.

\begin{enumerate}
	\item
	A new optimization framework is proposed to select the maximum number of SUs in the downlink and obtain power allocation for the selected SUs in order to satisfy their instantaneous rate requirements. The interference to PRs is kept below a specified threshold using margin parameters that compensate for CSI estimation errors. The proposed formulation is different from the formulations in \cite{yoo2006, huang2013, huang2012} which aim to maximize the sum-rate of the selected users and do not have interference constraints.	
	\item A new user selection and power allocation algorithm, called Delete-su-with-Maximum-Power-allocation (DMP), is proposed to obtain a feasible solution for the NP-hard optimization problem. Theoretical analysis of the algorithm is presented in order to compute the average number of SUs achieving the required rate, and average interference to primary receivers.	The proposed algorithm is shown to achieve near-optimal results if the number of antennas at SBS is an order of magnitude larger than the number of users.	
	\item The user selection algorithm in \cite{huang2012} is extended for application in an underlay CR setting. The extended algorithm, called Modified Delete-Minimum-Lambda (MDML), also uses margin parameters and is robust against imperfect CSI. The proposed DMP selection algorithm is shown to serve more users than MDML in an underlay CR network.
\end{enumerate}

\textit{Outline:} This paper is organized as follows. The system model and the optimization problem are presented in Section \ref{sec:Model_problem}. The DMP and MDML algorithms are presented in Section \ref{sec:solution}. Section \ref{sec:analysis_selection} presents the theoretical analysis and the optimality of the DMP algorithm. Simulation results are presented in Section \ref{sec:results}. Finally, the paper is concluded in Section \ref{sec:Conclusion}.

\textit{Notations:} We denote vectors by bold, lower-case letters, e.g., $\mathbf{h}$. Matrices are denoted by bold, upper case letters, e.g., $\mathbf{G}$. Scalars are denoted by non-bold letters e.g. $L$. Transpose, conjugate, and Hermitian of vectors and matrices are denoted by $(.)^T$, $(.)^*$, and $(.)^H$, respectively. The norm of a vector $\mathbf{h}$ is denoted by $||\mathbf{h}||$. $\Gamma(k,\theta)$ is the Gamma distribution with shape parameter $k$ and scale parameters $\theta$, whereas $\Gamma(x)$ is the Gamma function. The $i$-th element in the set $\mathcal{S}$ is denoted by $\mathcal{S}(i)$ and the cardinality of the set is denoted by $|\mathcal{S}|$. An empty set is denoted by $\emptyset$.


\section{System Model and Problem Formulation}
\label{sec:Model_problem}
\subsection{System Model}
Consider an underlay CR network with one SBS and $K$ SUs. The SBS is equipped with $M (\gg K)$ antennas. This network coexists with $L$ primary transmitter-receiver pairs. Let $\mathcal{T}=\{1,2,...,L\}$ be the set of primary transmitters (PTs) and  $\mathcal{R}=\{1,2,...,L\}$ be the set of PRs. The SUs, PRs and PTs are assumed to be single antenna terminals. Let $\mathbf{h}_{k} =\sqrt{\beta_k} \mathbf{\tilde{h}}_{k}  \in \mathbb{C}^{M \times 1}$ be the channel between SU-$k$ and the SBS where $\beta_k$ is the slow fading coefficient accounting for attenuation and shadowing and $\mathbf{\tilde{h}}_{k} \sim \mathcal{CN}(0,\mathbf{I})$ \cite{ngo2013a, chien2016, marzetta2010}. The channel between PT-$l$ and SU-$k$ is denoted by ${h_{lk}}=\sqrt{\beta_{lk}} {\tilde{h}_{lk}} \in \mathbb{C}, l\in \mathcal{T}, k=\{1,2,...,K\}$, and ${\tilde{h}_{lk}} \sim \mathcal{CN}(0,1)$. Similarly, the channel between PR-$l$  and the SBS is $\mathbf{h}_{l0}=\sqrt{\beta_{l0}} \mathbf{\tilde{h}}_{l0} \in \mathbb{C}^{M\times 1}, l\in \mathcal{R}, \mathbf{\tilde{h}}_{l0} \sim \mathcal{CN}(0,\mathbf{I})$. We consider a time-division duplex (TDD) systems and the channels are assumed to be reciprocal. { The network is depicted in Fig. \ref{fig:system_model}.}

\begin{figure}
	\centering
	\includegraphics[width=0.7\columnwidth]{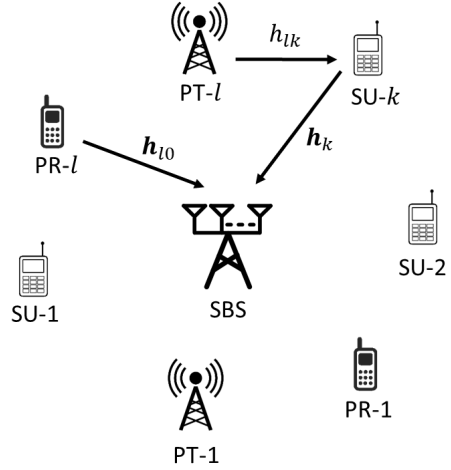}
	\caption{Network model showing channels between PT-$l$ and SU-$k$ ($h_{lk}$), PR-$l$ and the SBS ($\mathbf{h}_{l0}$), and SU-$k$ and the SBS ($\mathbf{h}_k$).}
	\label{fig:system_model}
	\vspace{-5mm}
\end{figure}

The SBS has imperfect knowledge of the channels $\mathbf{h}_{l0}, \mathbf{h}_{k}$. The estimates of channels are given by $\mathbf{\hat{h}}_{l0}=\mathbf{{h}}_{l0} + \mathbf{\Delta}_{l0}$ and $\mathbf{\hat{h}}_{k} = \mathbf{h}_{k} + \mathbf{\delta}_k$, respectively, where $\mathbf{\Delta}_{l0} \sim \mathcal{CN}(0,\sigma^2_{\Delta}\mathbf{I})$, and $\mathbf{\delta}_k \sim \mathcal{CN}(0,\sigma^2_{\delta}\mathbf{I})$ are the estimation errors.
We model the quality of CSI between primary and secondary system using $\sigma^2_\Delta = \frac{\sigma^2_w}{P_p}$, while the quality of CSI within the secondary system is modeled using $\sigma^2_\delta = \frac{\sigma^2_w}{P^0}$, where $P_p$ and $P^0$ are transmit powers from PTs and  the SBS, respectively, and $\sigma^2_w$ is the noise power at the SBS and SUs \cite{aquilina2015,razavi2014,maurer2011}. We consider block fading where the channels remain constant for a finite coherence interval.

Let $\mathcal{S}_0=\{1,2,...,K\}$ be the set of all SUs and $\mathcal{S} \subseteq \mathcal{S}_0$ be the set of SUs considered for downlink transmission and $P_k, k\in \mathcal{S}$ be the power allocated to SU-$k$ in the downlink when set ${\mathcal{S}}$ is selected. The ZF transmit vector for SU-$k$ depends on $\mathcal{S}$ and is denoted by {$\mathbf{v}^{\mathcal{S}}_k \in {\mathbb{C}^{M\times 1}}, k \in \mathcal{S}$}. The unit-norm ZF vectors are computed using \cite{huang2012}:
{
\begin{align}
\nonumber \mathbf{v}^{\mathcal{S}}_{\mathcal{S}(i)} &= \frac{\left[\mathbf{G_\mathcal{S}(G_\mathcal{S}^H G_\mathcal{S})^{-1}}\right]_i}{
||\left[\mathbf{G_\mathcal{S}(G_\mathcal{S}^H G_\mathcal{S})^{-1}}\right]_i||},
\\  \mathbf{G}_\mathcal{S} &= [\mathbf{\hat{h}}_{\mathcal{S}(1)}, \mathbf{\hat{h}}_{\mathcal{S}(2)},..., \mathbf{\hat{h}}_{\mathcal{S}(|\mathcal{S}|)}, \mathbf{\hat{h}}_{\mathcal{R}(1)0},...,\mathbf{\hat{h}}_{\mathcal{R}(L)0}],
\label{eq:zf_vectors}
\end{align}}
where $\mathcal{S}(i)$ is the $i$-th entry in $\mathcal{S}$, and $[\mathbf{A}]_i$ is the $i$-th column of matrix $\mathbf{A}$. { In the above equation, the matrix $\mathbf{G}_{\mathcal{S}} \in \mathbb{C}^{M \times (|\mathcal{S}| + L)}$ indicates the channel between SUs in set $\mathcal{S}$, PUs and the SBS.} It should be noted that the ZF vector $\mathbf{v}^{\mathcal{S}}_k$ is in the null-space of the estimated channels to SU-$j, j \in \mathcal{S}, j\neq k$. It is also in the null-space of the estimated channels to PRs. Therefore, $\mathbf{v}^{\mathcal{S}}_k$  satisfies $\mathbf{\hat{h}}^H_j \mathbf{v}^{\mathcal{S}}_k=0, j,k\in \mathcal{S}, j\neq k$ and $\mathbf{\hat{h}}^H_{l0}\mathbf{v}^{\mathcal{S}}_k=0, l \in \mathcal{R}$. 
 
The ZF beamforming vectors are not in the null-space of the true channel $\mathbf{h}^H_{l0}$ when the channel estimation error $\mathbf{\Delta}_{l0}$ is non-zero. Therefore, the interference caused at PR-$l, l\in \mathcal{R}$ is non-zero and can be expressed as:
 \begin{align}
I_l=\sum_{k \in S} I_{kl} = \sum_{k \in S} P_k |\mathbf{h}^H_{l0} \mathbf{v}^{\mathcal{S}}_k|^2, l \in \mathcal{R}, k \in \mathcal{S},
\label{eq:I_l}
 \end{align}
where $I_{kl}$ is the interference contribution of data stream of SU-$k$ towards PR-$l$. The interference $I_l$ depends on power allocation as well as the set of SUs selected. Similarly, the inter-SU interference at SU-$k$ due to the signal transmitted toward SU-$j$ can be expressed as
\begin{align}
I_{jk} = P_j |\mathbf{  h}^H_{k} \mathbf{v}^{\mathcal{S}}_j|^2, k,j \in \mathcal{S}, j \neq k.
\label{eq:I_jk}
\end{align}
Finally, the reverse interference at SU-$k$ is the sum of powers received from PTs: $I_{k} = \sum \limits_{l \in \mathcal{T}}P_p |{h_{lk}}|^2$, where $P_p$ is the power transmitted by PT-$l$.


\subsection{Optimization Problem}
\label{sec:problem}
Our goal is to select the maximum number of SUs for downlink transmission in order to satisfy specific instantaneous rate of $R^0_k$ to selected SUs, while keeping the interference towards PRs below $I^0$. The total available power at the SBS is $P^0$. Note that the estimated interference to PR-$l, l \in \mathcal{R}$ based on the estimated channel is $\hat{I}_l= \sum_{k \in \mathcal{S}} P_k |\mathbf{\hat{h}}^H_{l0} \mathbf{v}^{\mathcal{S}}_k|^2 = 0$. Since the true interference $I_l$ is non-zero, we add a margin parameter $\epsilon_1$ to define $\tilde{I}_l = \sum_{k \in \mathcal{S}} P_k (|\mathbf{\hat{h}}^H_{l0} \mathbf{v}^{\mathcal{S}}_k|^2 + \epsilon_1) $ as the new estimated value of the interference with margin.

Further, the instantaneous rate achieved at SU-$k$, when a set $\mathcal{S}$ is selected, is
\begin{align}
R^{\mathcal{S}}_k = \log_2 \left( 1 + \frac{P_k |\mathbf{ h}_{k}^H \mathbf{v}^{\mathcal{S}}_k|^2}{\sigma^2_w + I_k+ \sum \limits_{j \in \mathcal{S}, j \neq k} I_{jk} }\right), k \in \mathcal{S},
\label{eq:R_k_S}
\end{align}
where $\sigma^2_w$ is the noise power at the SU. Due to ZF beamforming, the estimated inter-SU interference will be zero, i.e., $\hat{I}_{jk} =P_j |\mathbf{\hat{h}}^H_k \mathbf{v}^\mathcal{S}_j|^2= 0$ due to $\mathbf{\hat{h}}^H_k \mathbf{v}^\mathcal{S}_j=0, k\neq j$. We use a margin parameter $\epsilon_2$ to compensate for the estimation error in the channels between SBS and SUs. Therefore, the estimated instantaneous rate with margin becomes
\begin{align}
\hat{R}^{\mathcal{S}}_k =  \log_2 \left( 1 + \frac{P_k |\mathbf{ \hat{h}}_{k}^H \mathbf{v}^{\mathcal{S}}_k|^2}{\sigma^2_w + {I}_{k} + \epsilon_2}\right), k \in \mathcal{S}.
\label{eq:R_k_S_est}
\end{align}

{ Note that, unlike $I_l$ and $I_{jk}$, the reverse interference term $I_k$ can be measured at SU-$k$ by observing combined signal received from all PUs during the channel estimation phase.} The optimization problem can then be formulated as follows:
\begin{align}
&\max_{\{\mathcal{S},P_k\}} |\mathcal{S}|
\label{eq:optim2_start}
\\ \text{Subject to}&:
 \tilde{I}_{l}=\sum_{k \in \mathcal{S}} P_k \left(|\mathbf{\hat{h}}^H_{l0} \mathbf{v}^{\mathcal{S}}_k|^2 + \epsilon_1 \right) \leq I^0,~ l \in \mathcal{R}
\label{eq:optim2_const1}
\\ &\hat{R}^{\mathcal{S}}_k  \geq R^0_k,~ k \in \mathcal{S},
\label{eq:optim2_const2}
\\ &\sum_{k\in \mathcal{S}} P_k \leq P^0, P_k \geq 0.
\label{eq:optim2_end}
\end{align}
Selection of parameters $\epsilon_1$ and $\epsilon_2$ is discussed in Section \ref{sec:selection_of_parameters}. By substituting $|\mathbf{\hat{h}}^H_{l0} \mathbf{v}^{\mathcal{S}}_k|^2 = 0$ and rearranging (\ref{eq:optim2_const2}), we obtain the following equivalent optimization problem:
\begin{align}
&\max_{\{\mathcal{S},P_k\}} |\mathcal{S}|
\label{eq:optim3_start}
\\\text{Subject to}&: \sum_{k\in \mathcal{S}} P_k \epsilon_1    \leq I^0,
\label{eq:optim3_const1}
\\ P_k &\geq \frac{(2^{R^0_k} -1)\left(\sigma^2_w + {I}_k + \epsilon_2 \right)}{|\mathbf{  \hat{h}}_k^H \mathbf{v}^{\mathcal{S}}_k|^2}, k \in \mathcal{S},
\label{eq:optim3_middle}
\\& \sum_{k\in \mathcal{S}} P_k \leq P^0.
\label{eq:optim3_end}
\end{align}
Using binary selection variables $s_k \in \{0,1\}$ to indicate whether SU-$k$ is selected ($s_k=1$) or not  ($s_k=0$), we can restate the above problem as:
\begin{align}
(\mathbf{P1}) \max_{\{s_k,P_k\}} &\sum_{k=1}^{K} s_k
\label{eq:optim4_start}
\\\text{Subject to}&:\sum_{k=1}^{K} s_k P_k  \leq \min \left( {I^0/\epsilon_1}, {P^0} \right),
\label{eq:optim4_const1}
\\ P_k &\geq  s_k \frac{(2^{R^0_k} -1)\left(\sigma^2_w +  {I}_k + \epsilon_2\right)}{|\mathbf{\hat{h}}_k^H \mathbf{v}^{\mathcal{S}}_k|^2}, k \in \mathcal{S}_0
\label{eq:optim4_const2}
\\s_k&=1,s_j=0, k \in \mathcal{S},  j \in \mathcal{S}_0 \backslash \mathcal{S}.
\label{eq:optim4_end}
\end{align}
{ The constraint (\ref{eq:optim4_const1}) is obtained by combining (\ref{eq:optim3_const1}) and (\ref{eq:optim3_end}). This constraint indicates that the power allocation is controlled by the interference limit $I^0$ if $I^0/\epsilon_1 \leq P^0$, while it is controlled by the transmit power limit $P^0$ if $I^0/\epsilon_1 > P^0$.}

\section{Selection algorithms and power allocation schemes}
\label{sec:solution}
The optimization problem (\ref{eq:optim4_start})-(\ref{eq:optim4_end}) to compute power allocations and selection variables is a non-convex mixed integer program and an NP-hard problem. Note that the computation of power allocations and selection variables depend on ZF vectors $\mathbf{v}^\mathcal{S}_k$ which in turn depend on the selected users. In order to solve this chicken-and-egg problem, we choose a particular set $\mathcal{S}$ and obtain ZF vectors and power allocation for that set. For a given selected set $\mathcal{S}$, the problem (\ref{eq:optim4_start})-(\ref{eq:optim4_end}) reduces to the following feasibility problem with power allocations as variables:
\begin{align}
	\text{Find}~~& P_k
	\label{eq:optim5_start}
	\\ \text{Subject to}&: \sum_{k\in \mathcal{S}} P_k  \leq \min \left( {I^0/\epsilon_1}, {P^0} \right),
	\label{eq:optim5_const1}
	\\ P_k &\geq  \frac{(2^{R^0_k} -1)\left(\sigma^2_w +  {I}_k + \epsilon_2\right)}{|\mathbf{\hat{h}}_k^H \mathbf{v}^{\mathcal{S}}_k|^2}, k \in \mathcal{S}.
	\label{eq:optim5_end}
\end{align}
If the power allocation
\begin{align}
P^{\mathcal{S}}_k =  \frac{(2^{R^0_k} -1)\left(\sigma^2_w +  {I}_k + \epsilon_2\right)}{|\mathbf{  \hat{h}}^H_k \mathbf{v}^{\mathcal{S}}_k|^2}, k \in \mathcal{S},
\label{eq:P_k_vup}
\end{align}
satisfies the constraint in (\ref{eq:optim5_const1}), then it provides the solution to the power allocation problem for the set $\mathcal{S}$. Note that the above power allocation attempts to satisfy specific instantaneous rate of SU-$k$. Therefore, it is referred to as Qos-Aware-Power-allocation.

Let $K^*$ be the cardinality of optimal sets. The problem (\ref{eq:optim4_start})-(\ref{eq:optim4_end}) can have multiple optimal sets, since multiple sets of the $K^*$ can satisfy the constraints (\ref{eq:optim4_const1})-(\ref{eq:optim4_end}). One approach of obtaining one of the optimal sets is to consider all possible sets of cardinalities $K, K-1,K-2,...,K^*$ one-by-one in decreasing order of cardinality, compute ZF vectors and power allocations by (\ref{eq:zf_vectors}) and (\ref{eq:P_k_vup}), respectively, and check whether the constraints in (\ref{eq:optim5_const1}) are satisfied. Such approach of user selection is prohibitively complex and impractical since the number of sets to be considered increases exponentially with $K$. As an example, for $K=20$ and $K^*=5$, the minimum number of sets to be considered are $\sum_{K'=K^*+1}^{K}$ ${K}\choose{K'}$ $\approx 1$ million. 
Therefore, there is a need to design a low-complexity algorithm to select users and obtain power allocations. 

As our goal is to maximize the cardinality of the set $\mathcal{S}$, we propose an approach which considers only one set of a particular cardinality that is obtained by dropping the SU that requires maximum power in a higher cardinality set. The selection algorithm  is initialized by selecting all the SUs, i.e., $\mathcal{S}= \mathcal{S}_0$. ZF vectors $\mathbf{v}^{\mathcal{S}}_k$ and power allocations $P^{\mathcal{S}}_k$ are computed for the set $\mathcal{S} = \mathcal{S}_0$ using (\ref{eq:zf_vectors}) and (\ref{eq:P_k_vup}), respectively. Then, the condition $\sum_{k \in \mathcal{S}} P^{\mathcal{S}}_k \leq \min(I^0/\epsilon_1, P^0)$ is checked. If the condition is not satisfied, the SU with maximum power allocation is dropped from the set and a set $\mathcal{S} = \mathcal{S} \backslash \{j\}$ of lower cardinality is considered, where $j=\arg \max_{k \in \mathcal{S}} P^{\mathcal{S}}_k$. The ZF vectors and power allocations are re-computed for the new set using (\ref{eq:zf_vectors}) and (\ref{eq:P_k_vup}), respectively. This process is continued until the constraint $\sum_{k \in \mathcal{S}} P^{\mathcal{S}}_k \leq \min(I^0/\epsilon_1, P^0)$ is satisfied. Since the SU with maximum power allocation is dropped in each iteration, the algorithm is called Delete-su-with-Maximum-Power (DMP). Note that dropping of the SU that requires the maximum power causes maximum reduction in $\sum_{k}s_k P_k$ in constraint (\ref{eq:optim4_const1}). This increases the probability that SUs included in set $\mathcal{S}\backslash \{j\}$ will satisfy the constraint (\ref{eq:optim4_const1}). The algorithmic steps are summarized in Algorithm \ref{alg:DMP}. { It should be noted that the SUs which require excess power to satisfy their rate requirements will not be selected by the DMP. For example, if $P^{\mathcal{S}}_k =  \frac{(2^{R^0_k} -1)\left(\sigma^2_w +  {I}_k + \epsilon_2\right)}{|\mathbf{  \hat{h}}^H_k \mathbf{v}^{\mathcal{S}}_k|^2} > \min(I^0/\epsilon_2, P^0)$, then SU-$k$ will not be selected.}

\begin{algorithm}[t]
	\caption{SU Selection Algorithm: DMP}
	\label{alg:DMP}	
	\begin{algorithmic}[1]
		\Statex { Input: channel estimates $\hat{\mathbf{h}}_k, \hat{\mathbf{h}}_{l0}$, reverse interference $I_k$, margins $\epsilon_1, \epsilon_2$, rate requirements $R^0_k$.}
		\State Select all SUs, i.e., $\mathcal{S}=\mathcal{S}_0$ and $s_k=1, \forall k$.
		\State Compute ZF vectors $\mathbf{v}^{\mathcal{S}_0}_k$ and power allocations $P^{\mathcal{S}_0}_k$ using (\ref{eq:zf_vectors}) and (\ref{eq:P_k_vup}), respectively.
		\While{$\mathcal{S}\neq \emptyset $}
		\If{$\sum_{k=1}^{K}s_k P^{\mathcal{S}}_k > \min \left( {I^0/\epsilon_1}, {P^0} \right)$}
		\State Remove SU with max $P^{\mathcal{S}}_k$:
		\label{step:drop}
		\Statex~~~~~~~~~~~~~~a. $j = \arg \max_{k \in \mathcal{S}} P^{\mathcal{S}}_k$.
		\Statex~~~~~~~~~~~~~~b. $\mathcal{S} \leftarrow \mathcal{S} \backslash \{j\}$, $s_{j}=0$.
		\State Update vectors and power allocations:
		\label{step:update}
		\Statex~~~~~~~~~~~~~~a. Compute $\mathbf{v}^{\mathcal{S}}_k$ for set $\mathcal{S}$ using (\ref{eq:zf_vectors}).
		\Statex~~~~~~~~~~~~~~b. Compute $P^{\mathcal{S}}_k$ using (\ref{eq:P_k_vup}).
		\Else~Stop.
		\EndIf
		\EndWhile
		\Statex { Output: set of selected SUs $\mathcal{S}_1^*=\mathcal{S}$, power allocations $P^{\mathcal{S}_1^*}_k = P^{\mathcal{S}}_k$.}
	\end{algorithmic}
\end{algorithm}

The selected set is denoted by $\mathcal{S}_1^*$ and the cardinality of the selected set is $K^*_1 = |\mathcal{S}^*_1|$. Due to imperfect CSI, all the selected SUs may not achieve the rate $R^0_k$. Therefore, we quantify the performance of the algorithm by $K^{**}_1 (<K^*_1)$ which is the number of SUs that achieve rate higher than $R^0_k$:
\begin{align}
K_1^{**} = \sum_{k \in \mathcal{S}^{*}_1} \mathbbm{1} (R^{\mathcal{S}^*_1}_k \geq R^0_k),
\end{align}
where $\mathbbm{1}(.)$ is the indicator function.

We also consider a low-complexity version of this algorithm where step \ref{step:update} of vector and power allocation update in Algorithm \ref{alg:DMP} is omitted. The set selected with this modified version without the vector is denoted by $\mathcal{S}_2^*$ and its cardinality by $K^{*}_2 = |\mathcal{S}_2^*|$. Further,  $K^{**}_2=\sum_{k \in \mathcal{S}^{*}_2} \mathbbm{1} (R^{\mathcal{S}^*_2}_k \geq R^0_k),$ is the number of SUs exceeding the required rate in DMP without vector update.

\begin{algorithm}[t]
	\caption{SU Selection Algorithm: MDML}
	\label{alg:MDML}	
	\begin{algorithmic}[1]
		\Statex { Input: channel estimates $\hat{\mathbf{h}}_k, \hat{\mathbf{h}}_{l0}$, reverse interference $I_k$, margins $\epsilon_1, \epsilon_2$.}
		\State Select all SUs, i.e., $\mathcal{S}=\mathcal{S}_0$ and $s_k=1, \forall k$.
		\State Compute $\mathbf{v}^{\mathcal{S}_0}_k, \lambda^{\mathcal{S}_0}_k, P^{\mathcal{S}_0}_k, \hat{R}(\mathcal{S}_0)$ using (\ref{eq:zf_vectors}), (\ref{eq:lambda_k}), (\ref{eq:P_k_wf}), and (\ref{eq:sum_rate}), respectively.
		\While{$\mathcal{S}\neq \emptyset$}	
		\State Delete SU with minimum lambda:
		\Statex~~~~~~~~~ $j = \arg \min_{k \in \mathcal{S}} \lambda^{\mathcal{S}}_k$.
		\Statex~~~~~~~~~ $\mathcal{S'} \leftarrow \mathcal{S}\backslash \{j\}$.
		\State Compute $\mathbf{v^{\mathcal{S'}}_k}, \lambda^{\mathcal{S'}}_k, P^{\mathcal{S'}}_k, \hat{R}(\mathcal{S'})$ using (\ref{eq:zf_vectors}), (\ref{eq:lambda_k}), (\ref{eq:P_k_wf}), and (\ref{eq:sum_rate}), respectively.
		\If{$\hat{R}(\mathcal{S'}) > \hat{R}(\mathcal{S})$}
		\Statex~~~~~~~~ $\mathcal{S} \leftarrow \mathcal{S'}, s_j =0$.
		\Else~ Stop				
		\EndIf
		\EndWhile
		\Statex { Output: set of selected SUs $\mathcal{S}_M^*=\mathcal{S}$, power allocations $P^{\mathcal{S}_M^*}_k = P^{\mathcal{S}}_k$.}
		
	\end{algorithmic}
\end{algorithm}
\subsubsection{MDML Selection Algorithm}
\label{sec:DML}
We extend the Delete-Minimum-Lambda (DML) selection scheme presented in \cite{huang2012} to underlay CR network with imperfect CSI. The DML algorithm in \cite{huang2012} selects users while maximizing the sum-rate of selected users. This approach does not take into account the rate constraints $R^{0}_k$ of SUs. It also does not take into account the imperfect CSI. Since DML was proposed for a primary massive MIMO network (CR network was not considered), it also does not include the reverse interference received at SU from primary transmitters. We modify the algorithm to include reverse interference and margin parameters for robustness against imperfect CSI. The Modified-DML (MDML) is described below.

In MDML, the equivalent channel gain between SU-$k, k \in \mathcal{S}$ and SBS is defined as:
\begin{align}
\lambda^{\mathcal{S}}_k = \frac{|\mathbf{  \hat{h}}_k^H \mathbf{v}^{\mathcal{S}}_k|^2}{\sigma^2_w +  {I}_k + \epsilon_2}, k \in \mathcal{S}.
\label{eq:lambda_k}
\end{align}
The power allocation for the SU included in the set $\mathcal{S}$ is obtained by water-filling. In order to satisfy the condition (\ref{eq:optim4_const1}), the maximum power level of $\min \left( {I^0/\epsilon_1}, {P^0} \right)$ is used to compute the power allocation by water-filling as shown below
\begin{align}
P^{\mathcal{S}}_k =  p^{\mathcal{S}}_k/ \lambda^{\mathcal{S}}_k,~
 p^{\mathcal{S}}_k = (\mu \lambda^{\mathcal{S}}_k -1)^{+},  k \in \mathcal{S},
\label{eq:P_k_wf}
\end{align}
where $(x)^+ = \max \{x,0\}$, and $\mu$ is the water level satisfying
\begin{align}
\sum_{k \in \mathcal{S}} \left(\mu - \frac{1}{\lambda^{\mathcal{S}}_k}\right) = \min \left( \frac{I^0}{\epsilon_1}, {P^0} \right).
\label{eq:water_level}
\end{align}
Further, the estimated sum rate for SUs included in $\mathcal{S}$ can be written as
\begin{align}
\hat{R}(\mathcal{S}) = \sum_{k \in \mathcal{S}} \log_2 \left(1+ P^{\mathcal{S}}_k \lambda^{\mathcal{S}}_k\right)
\label{eq:sum_rate}.
\end{align}
The MDML algorithm drops the SU with minimum $\lambda^{\mathcal{S}}_k$, if dropping the SU results in increased sum-rate. The algorithmic steps are summarized in Algorithm \ref{alg:MDML}. The selected set of SUs under this algorithm is denoted by $\mathcal{S}^*_M$.

\section{Analysis of DMP algorithm}
\label{sec:analysis_selection}
In this section, we provide analysis to compute $\mathbb{E}[K^{*}_1]$ and $\mathbb{E}[K^{**}_1]$ under fading channels $\mathbf{h}_k, {h_{lk}}$ and $\mathbf{h}_{l0}$. We also analyze the average interference to primary receivers. Since the coefficients $\beta_k, \beta_{lk}, \beta_{l0}$ change slowly over time, they are assumed to be constant in the analysis \cite{marzetta2010}.

\subsubsection{Average number of SUs served}
\label{sec:E_K_star}
The average number of SUs served using DMP is computed as follows:
\begin{align}
\mathbb{E}[K^{*}_1] = \sum \limits_{k=1}^{K} k \sum \limits_{\mathcal{S}\in \mathcal{S}_k } f(\mathcal{S})g(\mathcal{S}) ,
\end{align}
where $\mathcal{S}_k$ is the set of all sets of cardinality $k$, $f(\mathcal{S})$ is the probability of that the condition $\sum_{k\in S}P^{\mathcal{S}}_k \leq \min (I^0/\epsilon_1, P^0)$ is satisfied:
\begin{align}
f(\mathcal{S}) &= \Pr \left(\sum \limits_{k \in S} P^{\mathcal{S}}_k \leq \min \left(\frac{I^0}{\epsilon_1}, P^0\right) \right), 
\label{eq:f_S}
\end{align}
and $g(\mathcal{S})$ is the probability of arriving at set $\mathcal{S}$ during the algorithmic iterations. Since the set $\mathcal{S}_0$ is always considered, we have $g(\mathcal{S}_0)=1$ and $g(\mathcal{S}), \mathcal{S} \subset \mathcal{S}_0$ can be obtained using the following recursive expression:
\begin{align}	
g(\mathcal{S}) &= \sum_{j, j \notin \mathcal{S} } g(\mathcal{S} \cup \{j\}) P'\left(\{\mathcal{S} \cup \{j\}\} \backslash \{j\} \right),
\label{eq:g_S}
\end{align}
where $ P'\left(\{\mathcal{S} \cup \{j\}\} \backslash \{j\} \right) =  P'(\mathcal{S^+} \backslash \{j\})$ is the probability of dropping SU-$j$ from set $\mathcal{S^+} = \mathcal{S} \cup \{j\}$. This probability can be expressed as:
{\small \begin{align}
\nonumber P'(\mathcal{S}^+ \backslash \{j\} ) &=  (1-f(\mathcal{S}^+)) \Pr \left(P^{\mathcal{S}^+}_j > P^{\mathcal{S}^+}_1,..., P^{\mathcal{S}^+}_j > P^{\mathcal{S}^+}_{|\mathcal{S}^+|}\right),
\\&=(1-f(\mathcal{S}^+)) \int \limits_{0}^{\infty} p_{P^{\mathcal{S}^+}_j}(x)\prod\limits_{i \in \mathcal{S}^+, i\neq j}f_{P^{\mathcal{S}^+}_i}(x)dx,
\label{eq:P_dash_S}
\end{align}}
where $p_{P^{\mathcal{S}^+}_j}(x)$ and $f_{P^{\mathcal{S}^+}_i}(x)$ are the probability density function (pdf) of $P^{\mathcal{S}^+}_j$ and the cumulative distribution function (cdf) of $P^{\mathcal{S}^+}_i$, respectively. In order to evaluate (\ref{eq:f_S}), we need distributions of $P^\mathcal{S}_k$ which can be obtained from Theorem \ref{thm:pj_dist}. The distribution of $\sum_{k \in S}P^{\mathcal{S}_k}$ is required to evaluate (\ref{eq:P_dash_S}) which can be obtained from Corollary \ref{cor:sum_pj_dist}.

\begin{theorem}
	The power allocation $P^{\mathcal{S}}_k$ in (\ref{eq:P_k_vup}) is a Gamma random variable with shape and scale parameters $\kappa^p_k$ and $\theta^p_k$: $P^{\mathcal{S}}_k \sim \Gamma(\kappa^p_k, \gamma^{\mathcal{S}}_k \theta^p_k)$, where
	 \begin{align}
		\nonumber \kappa^p_k &=  \frac{\left(\sigma^2_w +\epsilon_2+ \sum_{l\in \mathcal{T}}P_p \beta_{lk}\right)^2}{ \sum_{l\in \mathcal{T}} (P_p\beta_{lk})^2},
		\\\theta^p_k &= \frac{\sum_{l\in \mathcal{T}}(P_p \beta_{lk})^2 }{\sigma^2_w +  \epsilon_2  +  \sum_{l\in \mathcal{T}} P_p\beta_{lk}}.
		\label{eq:thm1_1}
		\end{align}
	Similarly, for DMP without the vector update step, we have: $P^{\mathcal{S}_0}_k \sim \Gamma(\kappa^p_k, \gamma^{\mathcal{S}_0}_k \theta^p_k)$, where
	\begin{align}
		\nonumber \gamma^{\mathcal{S}}_k = \frac{2^{R^0_k}-1}{(\beta_k + \sigma^2_{\delta}) (M -|\mathcal{S}|-L+1)},
		\\ \gamma^{\mathcal{S}_0}_k = \frac{2^{R^0_k}-1}{(\beta_k + \sigma^2_{\delta}) (M -K-L+1)}.
		\label{eq:thm1_2}
		\end{align}
	\label{thm:pj_dist}
\end{theorem}
\begin{proof}		
	Appendix \ref{app_pj_dist}.	
\end{proof}


\begin{corollary}
	Sum of powers $\sum_{k \in \mathcal{S}} P^{\mathcal{S}}_k$ follows the Gamma distribution: $\sum_{k \in S} P^{\mathcal{S}}_k \sim \Gamma(\kappa_p, \theta_p)$, where
	\begin{align}
		 &\kappa_p = \frac{\left(\sum_{j \in \mathcal{S}} \kappa^p_j \gamma^{\mathcal{S}}_j \theta^p_j\right)^2}{\sum_{j \in \mathcal{S}} \kappa^p_j (\gamma^{\mathcal{S}}_j\theta^p_j)^2 },
		~\theta_p = \frac{\sum_{j \in \mathcal{S}} \kappa^p_j (\gamma^{\mathcal{S}}_j \theta^p_j)^2 }{\sum_{j \in \mathcal{S}} \kappa^p_j \gamma^{\mathcal{S}}_j \theta^p_j},
	\label{eq:k_P_theta_P}
		\end{align}
	\label{cor:sum_pj_dist}
\end{corollary}
\begin{proof}	
	Using Lemma 3 in \cite{hosseini2014}, the sum $\sum_{k \in \mathcal{S}} P^{\mathcal{S}}_k$ is modeled as a Gamma random variable with shape and scale parameters $\kappa_{p}$, and $\theta_{p}$, respectively, as defined in (\ref{eq:k_P_theta_P}). 
\end{proof}


\begin{corollary}
	Consider selection of two sets $\mathcal{S}_1$ and $\mathcal{S}_2$ containing SU-$k$. The power required to achieve rate $R^0_k$ at SU-$k$ with selection of $\mathcal{S}_1$ stochastically dominates the power required to achieve the same rate with the selection of $\mathcal{S}_2$, if $\mathcal{S}_2$ is subset of $\mathcal{S}_1$, i.e., $\Pr(P^{\mathcal{S}_1}_k \geq x)> \Pr(P^{\mathcal{S}_2}_k \geq x)$ for any $x$, if $\mathcal{S}_2 \subset \mathcal{S}_1$.
	\label{cor:dmpvup_P_k}
\end{corollary}
\begin{proof}
	Appendix \ref{app_cor_dmpvup_P_k}.
\end{proof}

\noindent\textit{Remark}: In DMP, if set $\mathcal{S}_1$ does not satisfy constraints in (\ref{eq:optim5_const1}), a subset of $\mathcal{S}_1$, say $\mathcal{S}_2$, is considered by dropping SU-$j$ that consumes the maximum power. Corollary \ref{cor:dmpvup_P_k} implies that the individual power requirements for all SUs still included in $\mathcal{S}_2$ reduce due to dropping of SU-$j$.

The expression for $\mathbb{E}[K^*_2]$ under DMP without vector update is obtained by evaluating $f(\mathcal{S})$ and $P'({S^+\backslash\{j\}})$ using distributions of $P^{\mathcal{S}_0}_k$ instead of $P^{\mathcal{S}}_k$.


\subsubsection{Average number of SUs exceeding the required rate}
\label{sec:E_K_starstar}
The average number of SUs achieving the minimum rate of $R^0_k$ using DMP can be expressed as follows:
\begin{align}
 \mathbb{E}[K^{**}_1] = \sum \limits_{k=1}^{K} \sum \limits_{\mathcal{S}: k \in \mathcal{S}} f(\mathcal{S}) g(\mathcal{S})  \Pr(R^{\mathcal{S}}_k \geq R^0_k),
\label{eq:E_K_starstar_1}
\end{align}
In order to compute the above expression, we need to compute the complementary cdf of $R^\mathcal{S}_k$ which is obtained by Theorem \ref{thm:R_k}.

\begin{theorem}
	The complementary cdf of the achieved rate at SU-k, if the set $\mathcal{S}$ is selected and the power is allocated by (\ref{eq:P_k_vup}), is given by:	
	\begin{align}
		\nonumber &\Pr(R^{\mathcal{S}}_k \geq y) =\\
		&\frac{1}{2\pi}\int \limits_{-\infty}^{\zeta_y} \int \limits_{-\infty}^{\infty}\left[(1-\theta^z_k jt)^{-\kappa^z_k} \prod \limits_{l \in \mathcal{T}} (1-\theta^z_{lk}jt)\right]~e^{-j2\pi wt}dt~dw,
		\label{eq:Pr_R_k}
		\end{align}
	where 
		\begin{align}
		\nonumber \zeta_y= C_y (\sigma^2_w + \epsilon_2) -\sigma^2_w , ~~
		C_y = \frac{\beta_k}{\beta_k + \sigma^2_\delta} \left(\frac{2^{R^0_k}-1}{2^y -1}\right),
		\end{align}
{		\small
	\begin{align}
	\nonumber \kappa^z_k&= \frac{\left(\sum \limits_{j \in \mathcal{S} \backslash \{k\}} {\gamma^{\mathcal{S}}_j \theta^p_j} {\Gamma(\kappa^p_j+1)}/{ \Gamma(\kappa^p_j)}\right)^2}{\sum \limits_{j \in \mathcal{S} \backslash \{k\}} \left({\gamma^{\mathcal{S}}_j \theta^p_j}/ {\Gamma(\kappa^p_j)}\right)^2 \left({2\Gamma(\kappa^p_j+2)\Gamma(\kappa^p_j) - \Gamma^2(\kappa^p_j+1)}\right)}, 
	\\\nonumber \theta^z_k&= \frac{\sigma^2_{\delta}}{\kappa^z_k} \sum \limits_{j \in \mathcal{S} \backslash \{k\}} \gamma^{\mathcal{S}}_j \theta^p_j \frac{{\Gamma(\kappa^p_j+1)} }{ { \Gamma(\kappa^p_j)}}, \text{and}
	\\\theta^z_{lk}&= (1-C_y)P_p\beta_{lk}
	\label{eq:kz_thetaz}
	\end{align}}
	\label{thm:R_k}
\end{theorem}
\begin{proof}	
	Appendix \ref{app_R_k}.
\end{proof}	
\noindent\textit{Remark}: If $\sigma^2_\delta= 0$, then we get $\Pr(R^\mathcal{S}_k \geq R^0_k)=1$. The proof is provided in Appendix \ref{app_R_k}.

The expression $\mathbb{E}[K^{**}_2]$ under DMP without vector update is obtained using the same expression as in the RHS of (\ref{eq:E_K_starstar_1}) by replacing  $\gamma^{\mathcal{S}}_j$ with $\gamma^{\mathcal{S}_0}_j$ in Theorems \ref{thm:pj_dist} and \ref{thm:R_k}.

\subsubsection{Average interference at PR-$l$}
\label{sec:I_l}
The expected value of the interference is computed as follows:
\begin{align}
\mathbb{E}[I_l] = \sum_{\mathcal{S}} g(\mathcal{S}) f(\mathcal{S})\sum_{k \in \mathcal{S}} \mathbb{E}\left[I_{kl}\right],~~ l \in \mathcal{R},
\label{eq:E_I_l}
\end{align}
where $\mathbb{E}\left[I_{kl}\right] = \mathbb{E}\left[P^{\mathcal{S}}_k |\mathbf{h}^H_{l0}\mathbf{v}^{\mathcal{S}}_k|^2\right] = \mathbb{E}\left[P^{\mathcal{S}}_k |\mathbf{\Delta}^H_{l0}\mathbf{v}^{\mathcal{S}}_k|^2\right]$. The second equality follows from $\mathbf{\hat{h}}^H_{l0} \mathbf{v}^{\mathcal{S}}_k =0$ due to zero forcing beamforming. The expression for $\mathbb{E}\left[I_{kl}\right]$ can be written using the distributions of $P^{\mathcal{S}}_k$ and $\mathbf{\Delta}_{l0}$ as follows:
\begin{align}
\mathbb{E}\left[I_{kl}\right] =\mathbb{E}\left[P^{\mathcal{S}}_k |\mathbf{\Delta}^H_{l0}\mathbf{v}^{\mathcal{S}}_k|^2\right] = \gamma^{\mathcal{S}}_k \theta^p_k \sigma^2_{\Delta} \frac{\Gamma(\kappa^p_k+1)}{\Gamma(\kappa^p_k)}.
\label{eq:I_kl_1}
\end{align}
The proof is shown in Appendix \ref{app_I_kl}. Similarly, the expression for $\mathbb{E}[I_l]$ in DMP without vector update is obtained by substituting the following in (\ref{eq:E_I_l}):
\begin{align}
\mathbb{E}\left[I_{kl}\right] = \gamma^{\mathcal{S}_0}_k \theta^p_k \sigma^2_{\Delta} \frac{\Gamma(\kappa^p_k+1)}{\Gamma(\kappa^p_k)}.
\label{eq:I_kl_2}
\end{align}
From (\ref{eq:E_I_l}), (\ref{eq:I_kl_1}), and (\ref{eq:I_kl_2}), we can see that the average interference to PRs is 0 for $\sigma^2_\Delta = 0$.

\subsection{{Optimality of DMP}}
\label{sec:optimality}
As described in Section \ref{sec:solution}, the power allocated to SU-$k$, $P^{\mathcal{S}_0}_k$, remains constant during the algorithmic iterations of DMP without vector update. Therefore, DMP without vector update effectively obtains the solution for the following problem:
\begin{align}
(\mathbf{P2})~~\max_{\{s_k\}} &\sum_{k=1}^{K} s_k
	\label{eq:optim6_start}
	\\\text{Subject to}&:\sum_{k=1}^{K} s_k P^{\mathcal{S}_0}_k  \leq \min \left( {I^0/\epsilon_1}, {P^0} \right),
	\label{eq:optim6_const1}
	\\s_k&=1,s_j=0, k \in \mathcal{S},  j \in \mathcal{S}_0 \backslash \mathcal{S}.
	\label{eq:optim6_end}
\end{align}
The solution obtained by DMP without vector update can be written as follows:
\begin{align}
	\mathcal{S}^*_2 = \arg \max_{\mathcal{S}: \sum_{k\in \mathcal{S}} P^{\mathcal{S}_0}_k  \leq \min \left( {I^0/\epsilon_1}, {P^0} \right)} |\mathcal{S}|.
\end{align}
The solution $\mathcal{S}^*_2$ is the optimal solution for $\mathbf{P2}$, since no set of higher cardinality can satisfy the constraint (\ref{eq:optim6_const1}) for fixed $P^{\mathcal{S}_0}_k$. This is because the proposed algorithm drops the SU with the highest power allocation in each iteration until the constraint (\ref{eq:optim6_const1}) is satisfied and addition of any SU to the set $\mathcal{S}^*_2$ will violate the constraint. 

Further, the optimal solutions of problems $\mathbf{P1}$ and $\mathbf{P2}$ differ due to the difference in power allocations $P^{\mathcal{S}^*}_k$ and $P^{\mathcal{S}_0}_k$, where $\mathcal{S}^*$ indicates the optimal solution set for $\mathbf{P1}$. The power allocations differ due to the difference in the number of nulls in the ZF vectors $\mathbf{v}^{\mathcal{S}^*}_k$ and $\mathbf{v}^{\mathcal{S}_0}_k$, which are denoted by $M-|\mathcal{S}^*|-L+1$ and $M-K-L+1$, respectively. The difference in the power allocations $P^{\mathcal{S}^*}_k$ and $P^{\mathcal{S}_0}_k$ becomes negligible if $M \gg K+L$. From Theorem \ref{thm:pj_dist}, we can also observe that the distribution of $P^{\mathcal{S}^*}_k$ approaches that of $P^{\mathcal{S}_0}_k$ as  $\gamma^{\mathcal{S}^*}_k \rightarrow \gamma^{\mathcal{S}_0}_k$, which occurs if $M \gg K+L$. Therefore, we can conclude that the problem $\mathbf{P1}$ becomes equivalent to $\mathbf{P2}$ for $M \gg K+L$ and  $|\mathcal{S}^*_2|$ approaches the optimal solution $|\mathcal{S}^*|$.

Finally, the number of SUs selected by the DMP algorithm is no less than the number of SUs selected by DMP without vector update, i.e, $|\mathcal{S}_2^*|\leq |\mathcal{S}_1^*|\leq |\mathcal{S}^*|$. This is because the power allocated to each SU in DMP is less that or equal to that in DMP without vector update, i.e., $P^{\mathcal{S}}_k \leq P^{\mathcal{S}_0}_k$ when a set $\mathcal{S}$ is selected. This can also be observed from Corollary \ref{cor:dmpvup_P_k}. The condition for selection of a set $\mathcal{S}$ under DMP {\small$\left(\sum \limits_{k \in \mathcal{S}} P^{\mathcal{S}}_k \leq \min \left(\frac{I^0}{\epsilon_1}, P^0\right)\right)$}  is always satisfied if the condition under DMP without vector update is satisfied {\small$\left(\sum \limits_{k \in \mathcal{S}} P^{\mathcal{S}_0}_k \leq \min \left(\frac{I^0}{\epsilon_1}, P^0\right)\right)$}, while the converse is not true. Therefore, we get $|\mathcal{S}_2^*|\leq |\mathcal{S}_1^*|\leq |\mathcal{S}^*|$. This phenomenon can be intuitively explained as follows. When the set $\mathcal{S}_0$ is selected initially, the number of degrees of freedom in the beamforming, after adding $L+K-1$ nulls, is $M-K-L+1$. The number of degrees of freedom increments if the ZF vectors are updated after dropping an SU. Therefore, the power requirements of SUs which are not dropped reduce as the ZF vectors $\mathbf{v}_k$ are better aligned with channels $\mathbf{\hat{h}}_k$. The reduction in power requirements implies that more SUs can be kept in the downlink, while satisfying the constraints of the problem $\mathbf{P1}$. Therefore, the number of SUs selected by the DMP is no less than the number of SUs selected by DMP without vector update.

\subsection{Selection of algorithm parameters}
\label{sec:selection_of_parameters}
The optimization framework $\mathbf{P1}$ in (\ref{eq:optim4_start})-(\ref{eq:optim4_end}) involves various parameters. In this section, we provide discussion on the selection of  these parameters. The parameters can be broadly classified into two categories: 1) network dependent fixed parameters: $P^0, I^0, R^0_k$, and 2) proposed margin parameters: $\epsilon_1, \epsilon_2$. 

\subsubsection{Network dependent parameters}
\label{sec: network_dependent_parameters}
The network dependent parameters are decided by the operators of secondary and primary networks. Consider, for example, that primary and secondary networks coexist in 3.5GHz band as CBRS users where primary system is Priority Access License (PAL) user and secondary system is General Access Authorization (GAA) user \cite{fcc_15_47A1, ye2016b}. The value of $P^0$ will be determined using the power amplifier used at the SBS. Typical value of $P^0=40$dBm is used for BS under sub-6GHz bands. The rate constraints $R^0_k$ are determined by the operator of the secondary network depending on the QoS requirements for the SUs. The interference constraint $I^0$ is determined by the operator of primary network. The value of $I^0$ determines the SINR degradation of PUs due to the coexisting SUs. For example, if SINR degradation of $<1$dB is desired then $I^0$ should be set such that $I^0/\sigma^2_w< -6$dB, where $\sigma^2_w$ is the noise power at the PU.

\subsubsection{Proposed margin parameters}
\label{sec: proposed_margin_parameters}
Once the values of $P^0, I^0$ and $R^0_k$ are fixed, the algorithm specific margin parameters $\epsilon_1 $ and $\epsilon_2$ are set as described next. Margin parameters $\epsilon_1$ and $\epsilon_2$ are used in the proposed optimization framework in order to compensate for interference to PRs and inter-SU interference, respectively, resulting form imperfect CSI estimates. 

In order to select appropriate value of $\epsilon_1$, let us consider average value of the true interference $I_l$ at PR-$l$ for given channel estimates $\mathbf{\hat{h}}_{l0}$ and a selected set $\mathcal{S}$:
\begin{align}
\nonumber \mathbb{E}[I_l | \mathbf{\hat{h}}_{l0}, \mathcal{S}] &= \sum_{k \in \mathcal{S}} P^\mathcal{S}_k \mathbb{E}[|\mathbf{h}^H_{l0} \mathbf{v}^{\mathcal{S}}_k|^2  | \mathbf{\hat{h}}_{l0}, \mathcal{S}], l \in \mathcal{R}, k \in \mathcal{S},\\
 &= \sum_{k \in \mathcal{S}} P^\mathcal{S}_k \sigma^2_\Delta.
\end{align}
The last equality in the above equation follows from the fact that beamforming vectors $\mathbf{v}^\mathcal{S}_k$ are unit vectors that are in nullspace of estimated channels $\mathbf{\hat{h}}^H_{l0}$. Further, the proposed optimization problem ensures that $\sum_{k \in S} P^\mathcal{S}_k \epsilon_1 \leq I^0$ due to constraint in (\ref{eq:optim3_const1}). Therefore, the average interference $\mathbb{E}[I_l | \mathbf{\hat{h}}_{l0}, \mathcal{S}]$ is below the threshold $I^0$ if $\epsilon_1 \geq \sigma^2_\Delta$.

Now, let us consider the selection of $\epsilon_2$. From (\ref{eq:R_k_S}) and (\ref{eq:R_k_S_est}), we can see that the variable $\epsilon_2$ serves as a placeholder for inter-SU interference $\sum_{j \in \mathcal{S}, j\neq k}I_{jk}$. The variable $\epsilon_2$ ensures that higher power is allocated to SU-$k$ to compensate for the inter-SU interference as seen from (\ref{eq:P_k_vup}). Therefore, the value of $\epsilon_2$ should be selected such that $\epsilon_2 \geq \sum_{j, j\neq k}I_{jk}$. However, the instantaneous value of inter-SU interference is unknown. We propose to set the parameter value such that it exceeds the expected value of inter-SU interference for given channel estimates $\mathbf{\hat{h}}_{j}, \mathbf{\hat{h}}_{k}$ and a selected set $\mathcal{S}$:
\begin{align}
\nonumber \epsilon_2 &\geq \mathbb{E}\left[\sum_{j \in \mathcal{S}, j\neq k}I_{jk} \bigg | \mathbf{\hat{h}}_{j},\mathbf{\hat{h}}_{k}, \mathcal{S}\right],\\
\nonumber & = \mathbb{E}\left[\sum_{j \in {\mathcal{S}}, j\neq k}P^{\mathcal{S}}_j |\mathbf{{h}}^H_k \mathbf{v}^{\mathcal{S}}_j|^2  \bigg | \mathbf{\hat{h}}_{j},\mathbf{\hat{h}}_{k}, \mathcal{S}\right], k,j \in{\mathcal{S}}, k\neq j,\\
& \stackrel{(a)}{=} \sum_{j \in {\mathcal{S}}, j\neq k}P^{\mathcal{S}}_j \mathbb{E}[|\delta_k^H \mathbf{v}^{\mathcal{S}}_j |^2],\\
& \stackrel{(b)}{=} \sum_{j \in {\mathcal{S}}, j\neq k}P^{\mathcal{S}}_j \sigma^2_\delta.
\label{eq:epsilon_2_final}
\end{align}
The equality $(a)$ results due to the fact that $\mathbf{\hat{h}}_k \mathbf{v}^{\mathcal{S}}_j = 0$. The equality $(b)$ results from $\delta_k \sim CN(0,\sigma^2_\delta \mathbf{I})$ and $||\mathbf{v}^{\mathcal{S}}_k||=1$. Thus, $\epsilon_2$ should be greater than $\sum_{j \in {\mathcal{S}}, j\neq k}P^{\mathcal{S}}_j \sigma^2_\delta$. However, the power allocations $P^{\mathcal{S}}_j$ are not known in advance. Therefore, we set $\epsilon_2\geq P^0 \sigma^2_\delta$, which ensures that the condition in (\ref{eq:epsilon_2_final}) is satisfied since $P^0 \sigma^2_\delta \geq \sum_{j \in {\mathcal{S}}, j\neq k}P^{\mathcal{S}}_j \sigma^2_\delta$.

\subsubsection{Optimum margins $\epsilon_1, \epsilon_2$}
\label{sec: optimum_margins}
We observe that larger $\epsilon_1$ in (\ref{eq:optim5_const1}) results in admitting fewer SUs in the downlink. Similarly, larger $\epsilon_2$ results in larger power allocation according to (\ref{eq:P_k_vup}), further resulting in dropping of SUs due to the constraint $\sum_{k\in \mathcal{S}} P_k \leq \min(I^0/\epsilon_1, P^0)$ in (\ref{eq:optim5_const1}). In order to satisfy the rate and interference constraints while admitting maximum number of SUs in the downlink, it is necessary to set $\epsilon_1$ and $\epsilon_2$ to the smallest possible values. Therefore, the setting $\epsilon_1=\sigma^2_\Delta$ and $\epsilon_2 = P^0 \sigma^2_\delta$ results in serving maximum number of SUs with given interference and rate constraints.

\subsection{{Complexity Analysis}}
\label{sec:complexity}
Computational complexity of DMP as well as MDML is dominated by the computation of ZF vectors. For a set $\mathcal{S}$, the complexity of obtaining the ZF vectors is $\mathcal{O}\left(M (|\mathcal{S}| +L)^3\right)$ \cite{huang2012}. Since ZF vectors are updated in each iteration of DMP and MDML until a feasible set is reached, the worst case complexity of both algorithms is $\mathcal{O}\left(M K (K+L)^3 \right)$, while the worst case complexity of DMP without vector update is $\mathcal{O}\left(M (K+L)^3\right)$. 

As shown in the previous section, the solution $|\mathcal{S}^*_2|$ obtained by DMP without vector update approaches the optimal value with large $M$. Therefore, we can conclude that near-optimal number of SUs can be selected by the proposed algorithm while reducing the computational complexity by a factor of $K$ as compared to MDML.

\section{Simulation Results}
\label{sec:results}
In the results shown below, the noise power $\sigma^2_w$ is assumed to be $-100$dBm, the transmitted power from primary transmitters is $P_p=20$dBm, the transmit power limit is $P^0=40$dBm. The variance of error is modeled assuming reciprocal channels in a time-division duplex system as $\sigma^2_\Delta = \sigma^2_w/P_p$ and $\sigma^2_\delta = \sigma^2_w/P^0$ \cite{aquilina2015,razavi2014,maurer2011}. 
We consider uniformly distributed SUs and primary transmitters and receivers in a circular cell of radius $2$km with the SBS at the center. The minimum distance between the SBS and SUs is $100$m \cite{ngo2013, ngo2013a}. For each realization of locations, the slow fading coefficients between two nodes are computed as $\beta = \rho d^{-3.8}$, where $d$ is the distance between the two nodes and $\rho$ is a log-normal shadowing variable with standard deviation $\sigma_{s} = 8$dB. The margin parameters are set as $\epsilon_1 =  \sigma^2_\Delta$ and $\epsilon_2 = P^0 \sigma^2_\delta$. We simulate the algorithms for 1000 realizations of the channel for each realization of locations. Analytical and simulation results are averaged over 1000 realizations of the locations.

\begin{figure}
	\centering
	\includegraphics[width=\columnwidth]{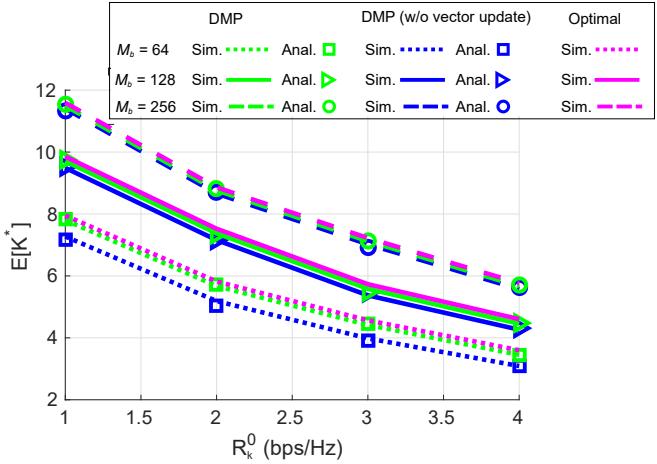}
	\caption{\small Comparison of number of SUs selected by DMP and optimal selection. $L=4$, $K=20$, $I^0 = -106$ dBm.}
	\label{fig:opt}
\end{figure}

\textit{{Comparison with optimal solution:}}
The comparison of the average number of SUs selected by DMP and optimal selection is shown in Fig. \ref{fig:opt}. The optimal solution is obtained by considering all possible sets of cardinalities $K, K-1,K-2,...,K^*$ one-by-one in decreasing order of cardinality, computing ZF vectors and power allocations by (\ref{eq:P_k_vup}), until the constraints in (\ref{eq:optim5_const1}) are satisfied. We observe that the number of SUs selected by DMP is very similar to that by optimal selection and the difference between $E[|\mathcal{S}^*_1|] \approx E[|\mathcal{S}^*|]$. As the number of antennas increased from 64 to 256, the difference between the performance of DMP without vector update and optimal selection becomes negligible as explained in Section \ref{sec:optimality}.

\textit{Impact of $R^0_k$:} The proposed DMP algorithm is designed to satisfy the minimum rate required by the SUs unlike MDML which does not take into account the rate requirements. Therefore, it can be observed that the DMP serves more SUs exceeding the minimum required rate than MDML when $R^0_k$ is uniformly distributed in the range $(0,4]$ as seen in Fig. \ref{fig:diff_R0}. The performance of the DMP and MDML becomes similar as the rate requirements increase to $4$bps/Hz. Further, it can be observed that the performance curves of the DMP and the MDML converge at a higher rate for large number of antennas. This indicates that the performance gain obtained by the DMP over MDML increases for a given rate requirement as the number of antennas increase.


\begin{figure}[t!]
	\centering	
	\begin{subfigure}[b]{\columnwidth}
		\includegraphics[width=\columnwidth]{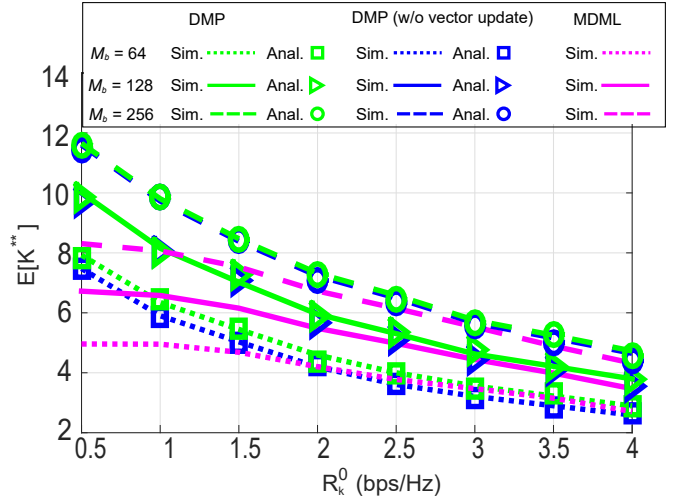}
	\end{subfigure}		
	\caption{\small Impact of rate constraints. $L=4$, $K=20$, $I^0 = -106$ dBm.}
	\label{fig:diff_R0}
\end{figure}

\textit{Impact of $I^0$:} The interference threshold $I^0$ limits the total transmitted power below $I^0/\epsilon_1$, thereby limiting the number of SUs served in both DMP and MDML. It should be noted that the interference of $-100$, $-106$ and $-110$ dBm results in SINR loss of $3, 0.97$ and $0.41$dB, respectively at primary receivers. As shown in  Fig. \ref{fig:diff_I0}, the number of SUs served by the three algorithms increases by 1.5 times with increased interference threshold from $-110$dBm to $-100$dBm at the cost of reduced { signal-to-interference-plus-noise ratio (SINR)} at primary receivers. It can be observed that the performance gain obtained by the DMP over MDML is consistent for different interference thresholds. 

\begin{figure}[t]
	\centering	
	\begin{subfigure}[b]{\columnwidth}
		\includegraphics[width=1\columnwidth]{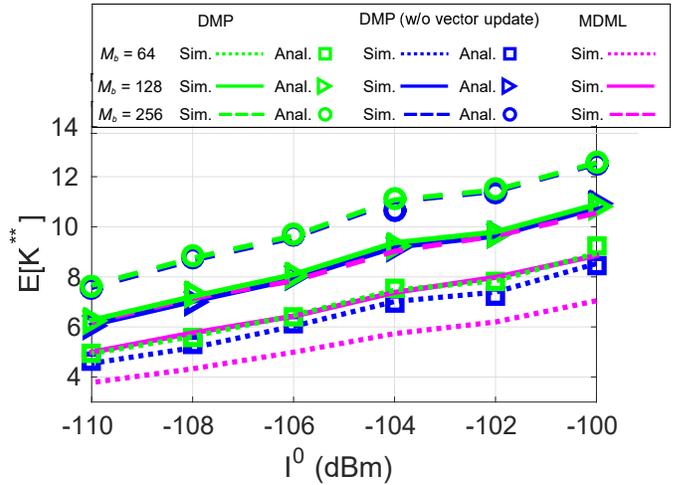}
	\end{subfigure}		
	\caption{\small Impact of interference constraints. $L=4$, $K=20$. $R^0_k=1$ bps/Hz.}
	\label{fig:diff_I0}
		\vspace{-4mm}
\end{figure}

\textit{Impact of number of primary tx-rx pairs $L$:} Increased number of PTs in the network increases the reverse interference $I_k$ to SUs. This results in increased power requirement $P_k$ for the SU-$k$ to satisfy the rate according to (\ref{eq:P_k_vup}). This increased power requirement in turn results in dropping of more SUs in step 6 of the DMP algorithm. Therefore, the number of SUs served by the proposed algorithm reduces with higher $L$ as shown in Fig. \ref{fig:diff_Lp}.

\begin{figure}[t!]
	\centering	
	\begin{subfigure}[b]{\columnwidth}
		\includegraphics[width=1\columnwidth]{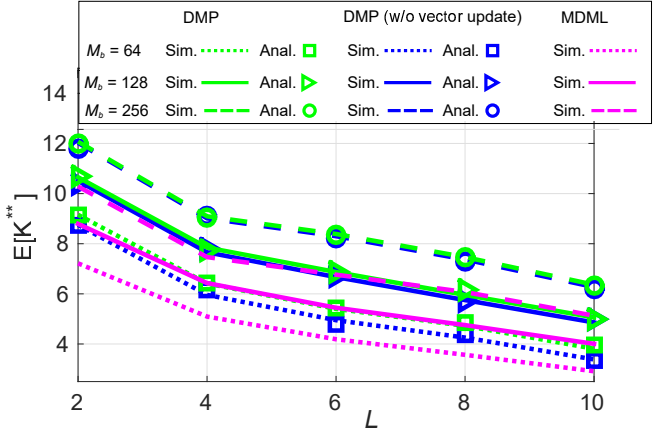}
		\label{fig:diff_Lp_E_k_starstar}
	\end{subfigure}		
	\caption{\small Impact number of primary pairs. $K=20$,$I^0 = -106$ dBm. $R^0_k=1$ bps/Hz.}
	\label{fig:diff_Lp}
\end{figure}

\textit{Impact of total number of SUs:}
The impact of increasing number SUs is shown Fig. \ref{fig:diff_K_unif_R0}. The number of SUs exceeding the required rate increases almost linearly for $M=128$ and $M=256$ under DMP when the rate requirements $R^0_k$ are uniformly distributed in the range $ (0,4]$ bps/Hz. 
We also observe that the difference in the performance of the DMP with and without vector update reduces with increased number of antennas. This is due to the fact that the ratio $\gamma^{\mathcal{S}_0}/\gamma^{\mathcal{S}^*_1}$ is close to one which results in similar power allocations for SUs with and without vector update, thereby resulting in similar number of SUs being dropped in the step \ref{step:drop} of the DMP algorithm.

\begin{figure}[t!]
	\centering	
	\begin{subfigure}[b]{\columnwidth}
		\includegraphics[width=1\columnwidth]{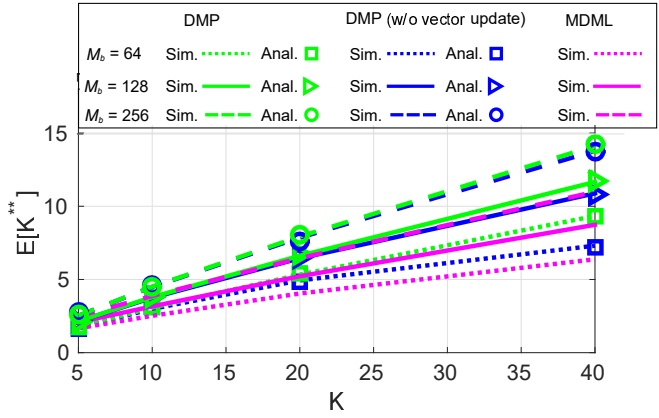}
		\label{fig:diff_K_unif_R0_E_k_starstar}
	\end{subfigure}		
	\caption{\small Impact of total number of SUs $K$ in the network. $L=4$, $I^0 = -106$ dBm.  $R^0_k$ is uniformly distributed in (0,4] bps/Hz.}
	\label{fig:diff_K_unif_R0}
\end{figure}

\textit{Optimality of margins:} The margin parameters $\epsilon_1$ and $\epsilon_2$ are used to protect the PRs from the interference under imperfect CSI. In order to study the impact of margins on the performance, we plot the average interference and the number of SUs served for different values of $\epsilon_1$ and $\epsilon_2$. As shown in Fig. \ref{fig:diff_eta_ep_unif_R0_avg_int}, the average interference remains below the threshold for $\epsilon_1>\sigma^2_\Delta$. This result holds for values of $\epsilon_2/ P^0 \sigma^2_\delta$ in range [0,4], because the variable $\epsilon_2$ does not significantly affect the average interference. As mentioned in Section \ref{sec:selection_of_parameters}, larger values of $\epsilon_1$ result in smaller the number of SUs served. Therefore, we keep $\epsilon_1=\sigma^2_\Delta = \sigma^2_w/ P_p$ and plot $E[K^{**}]$ for range of values of $\epsilon_2$, as shown in Fig. \ref{fig:diff_eta_ep_unif_R0_E_k_starstar}. For $\epsilon_2< P^0 \sigma^2_\delta$, fewer SUs receive the required rate due to inter-SU interference, while for $\epsilon_2 > P^0 \sigma^2_\delta$ fewer SUs are admitted in the downlink due to large power allocation. Therefore, we see that the maximum number of SUs are served for $\epsilon_2=P^0 \sigma^2_\delta$ as described in Section \ref{sec: optimum_margins}.


\section{Conclusion}
\label{sec:Conclusion}
In this paper, we proposed an optimization framework in order to serve the maximum number of SUs in an underlay CR network consisting of a secondary BS equipped with a large number of antennas. The proposed framework uses margin parameters to limit the interference to PUs below a specified threshold under imperfect knowledge of CSI. A new user selection and power allocation algorithm, referred to as DMP, is proposed that is based on ZF beamforming and power allocation that satisfies specific rate requirements of selected SUs. Theoretical analysis is presented to compute the number of SUs selected and the interference caused at PUs by the proposed algorithm. Results show that the proposed DMP algorithm serves more SUs than modified DML algorithm for lower rate requirements. As the rate requirements for the SUs increase, the performance of the modified DML algorithm approaches that of DMP. A low complexity version of DMP without ZF vector update is also studied. This algorithm reduces the complexity by a factor of the number of SUs and provides similar performance as DMP with vector update when the number of SBS antennas is an order of magnitude larger than the number of SUs. The analysis and simulation results show that the number of SUs selected by the proposed algorithm approaches the optimal solution if the number of SBS antennas is an order of magnitude larger than the number of SUs and PUs in the network. 

\begin{figure}[t!]
	\centering	
	\begin{subfigure}[b]{0.9\columnwidth}
		\includegraphics[width=\columnwidth]{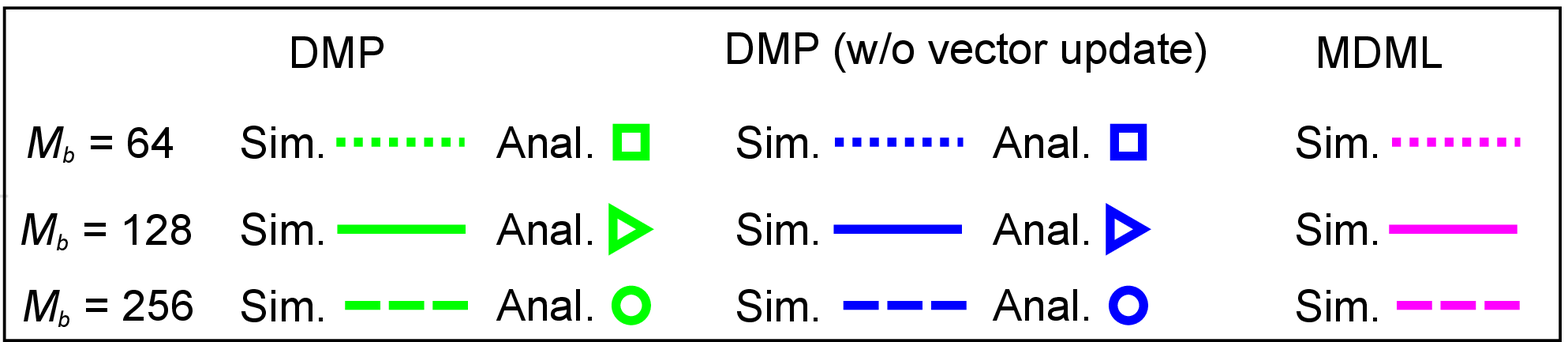}
		\label{fig:legend_diff_eta_ep_unif_R0}
		\vspace{-4mm}		
	\end{subfigure}			
	\\
	\begin{subfigure}[b]{\columnwidth}
		\includegraphics[width=\columnwidth]{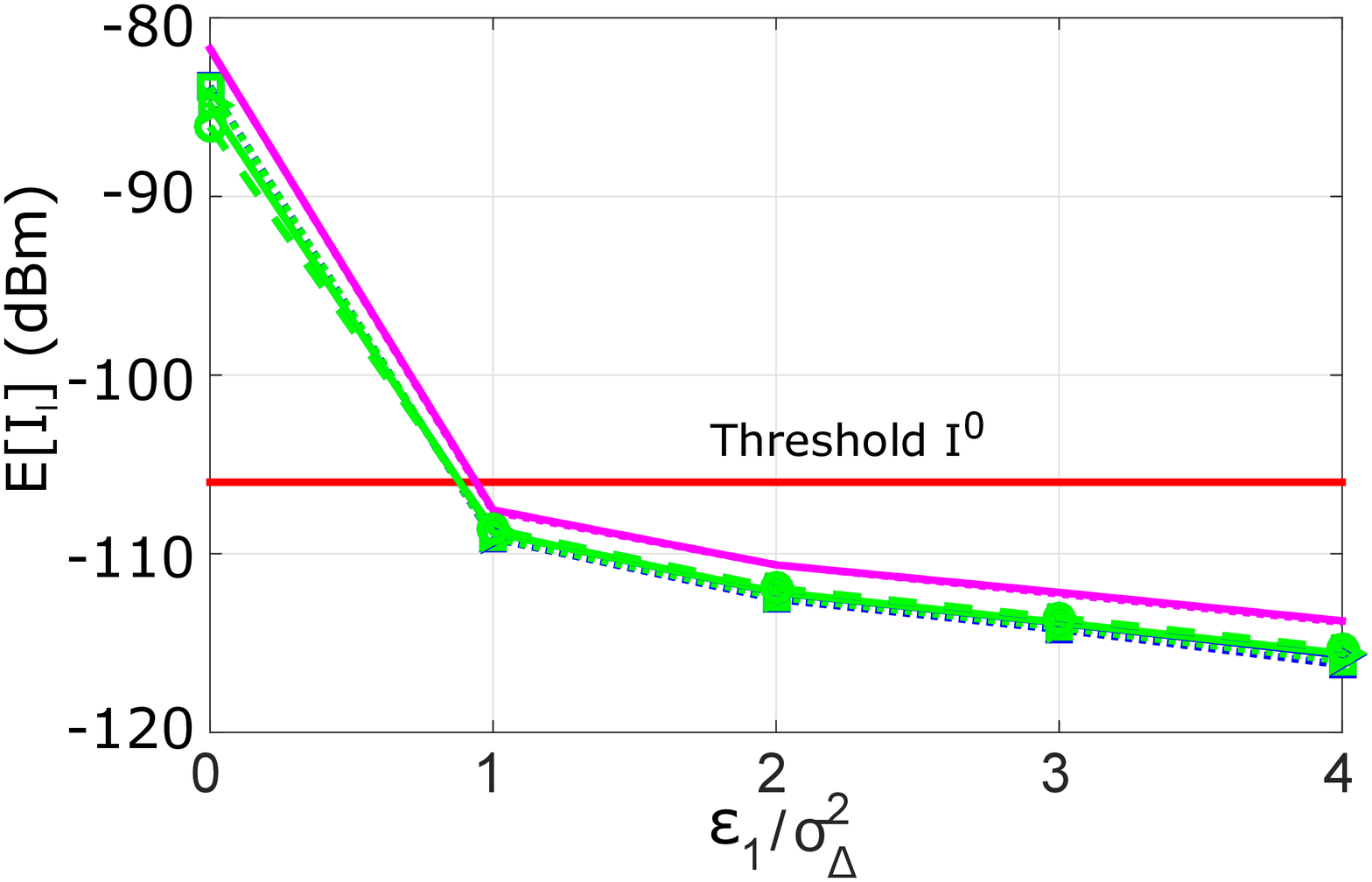}
		\caption{\small Avg. interference to primary receivers. Results hold for $\frac{\epsilon_2}{P^0 \sigma^2_\delta} \in  [0,4]$.}
		\label{fig:diff_eta_ep_unif_R0_avg_int}
	\end{subfigure}	
	\begin{subfigure}[b]{\columnwidth}
		\includegraphics[width=\columnwidth]{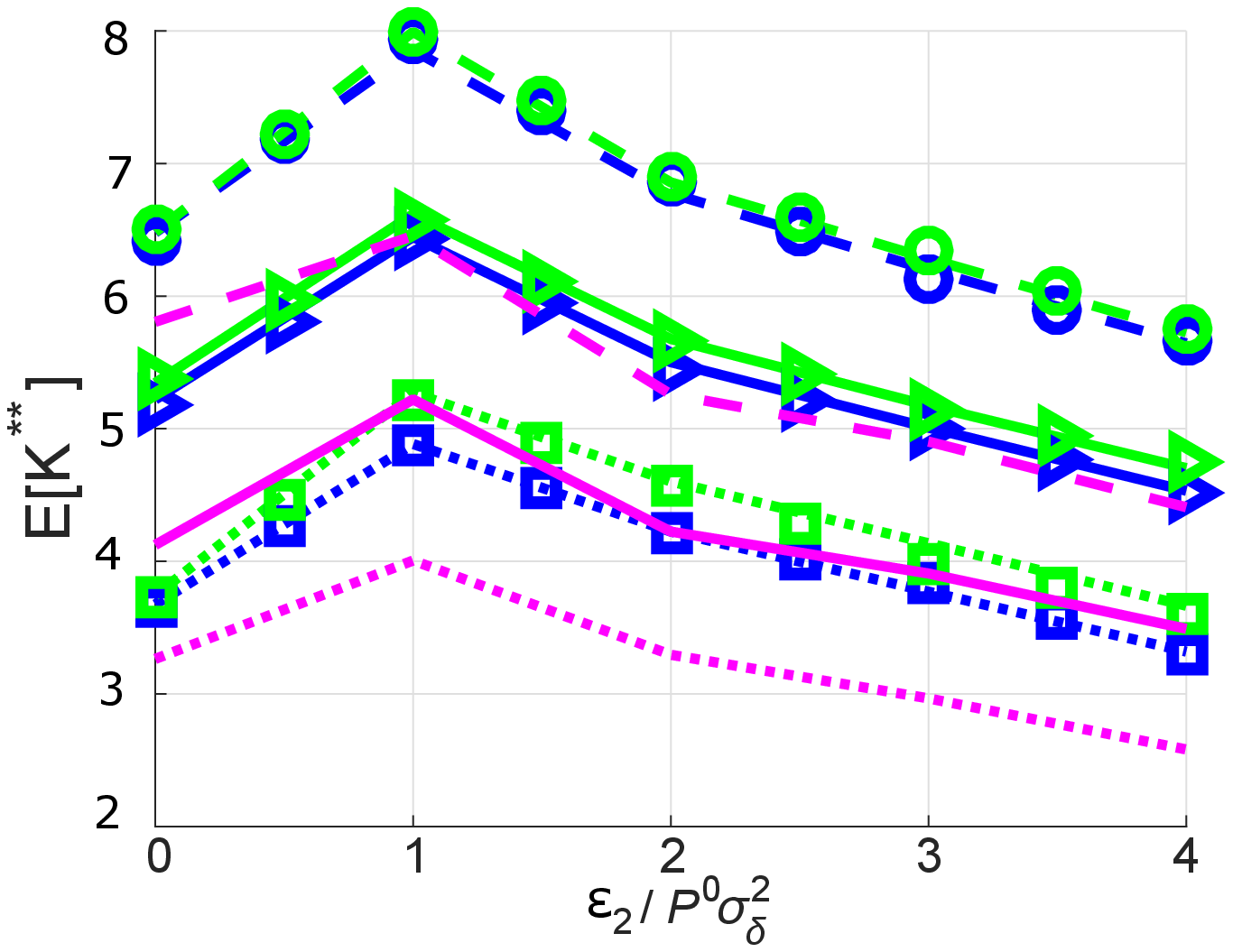}
		\caption{\small Average \# SUs served with minimum rate $R^0_k$ for $\epsilon_1 = \sigma^2_\Delta = \sigma^2_w / P_p$.}
		\label{fig:diff_eta_ep_unif_R0_E_k_starstar}
	\end{subfigure}			
	\caption{\small Impact of margin parameters. $K=20$, $I^0 = -106$ dBm. $R^0_k$ is uniformly distributed in (0,4] bps/Hz.}
	\label{fig:diff_eta_ep_unif_R0}
\end{figure}

\appendices

\section{Derivation of distributions of $P^{\mathcal{S}}_k$ and $P^{\mathcal{S}_0}_k$}
\label{app_pj_dist}
From (\ref{eq:P_k_vup}), $P^{\mathcal{S}}_j$ can be expressed as
\begin{align}
P^{\mathcal{S}}_k =  \frac{(2^{R^0_k} -1)\left(\sigma^2_w + {I}_k+\epsilon_2\right)}{|\mathbf{  \hat{h}}^H_k \mathbf{v}^{\mathcal{S}}_k|^2} = (2^{R^0_k}-1)\frac{X}{Y},
\end{align}
where $X = \sigma^2_w + {I}_k + \epsilon_2$ and $Y = |\mathbf{\hat{h}}^H_k \mathbf{v}^{\mathcal{S}}_k|^2 = |({\mathbf{h}^H_k + \mathbf{\delta}^H_k}) \mathbf{v}^{\mathcal{S}}_k|^2$. The vector $\mathbf{h}^H_k + \mathbf{\delta}^H_k \sim CN(0, \beta_k +\sigma^2_\delta)$ is an isotropic vector, while the vector $\mathbf{v}^{\mathcal{S}}_k$ spans $M-|\mathcal{S}|-L+1$ dimensional space due to $|\mathcal{S}|-1+L$ nulls. Therefore, $Y$ is modeled as a Gamma random variable with shape and scale parameters $M - |\mathcal{S}|-L+1$ and $\beta_k + \sigma^2_\delta$, respectively \cite[lemma 1]{hosseini2014} \cite{chaudhari2017a}, i.e., $Y \sim \Gamma(M - |\mathcal{S}|-L+1, \beta_k + \sigma^2_\delta)$. For simplicity of the analysis, we approximate $Y$ with its average value: $Y=(M - |\mathcal{S}|-L+1)(\beta_k + \sigma^2_\delta)$. This approximation is valid because the variance of $Y$, $(M - |S|-L+1)(\beta_k + \sigma^2_\delta)^2$, is negligible as compared to its mean  $(M - |\mathcal{S}|-L+1)(\beta_k + \sigma^2_\delta)$.

Further, consider the variable $X = \sigma^2_w + {I}_{k} +\epsilon_2$. We define a constant $C=\sigma^2_w + \epsilon_2$. Each term in the summation ${I}_k = \sum_{l \in \mathcal{T}}P_p |{{h}_{lk}}|^2 $ is modeled as a Gamma random variable with distribution $\Gamma(1,P_p\beta_{lk})$. Therefore, the mean of $X$ is $C+ \sum_{l\in \mathcal{T}}P_p \beta_{lk}$, while its variance is $\sum_{l\in \mathcal{T}}{P_p}^2\beta_{lk}^2$. We model $X$ as a Gamma random variable with size and shape parameters $\kappa^p_k$ and $\theta^p_k$, respectively, i.e., $X \sim \Gamma(\kappa^p_k, \theta^p_k)$ Therefore, we have
\begin{align}
\nonumber \mathbb{E}[X] = \kappa^p_k\theta^p_k= C+\sum_{l\in \mathcal{T}} P_p \beta_{lk} ,
\end{align}
\begin{align}
\mathbf{var}[X] = \kappa^p_k(\theta^p_k)^2=  \sum_{l\in \mathcal{T}}{P_p}^2\beta_{lk} ^2.
\label{eq: app_kPk_theta_Pk}
\end{align}
By solving for $\kappa^p_k$ and $\theta^p_k$, we obtain (\ref{eq:thm1_1}). 

Finally, $P^{\mathcal{S}}_k = \frac{2^{R^0_k}-1}{(M - |S|-L+1)(\beta_k + \sigma^2_\delta)}  X = \gamma^{\mathcal{S}}_k X$ is a Gamma random variable with size and shape parameters $\kappa^p_k$ and $\gamma^{\mathcal{S}}_k \theta^p_k$, respectively, i.e., $P^{\mathcal{S}}_k \sim \Gamma(\kappa^p_k, \gamma^{\mathcal{S}}_k \theta^p_k)$. Note that $P^{\mathcal{S}}_k, k=1,2,...$ are independent variables since they are functions of independent random variables ${{h}_{lk}}$.

The distribution of $P^{\mathcal{S}_0}_k$ is obtained by following the above derivation with $Y=(M-K-L+1)(\beta_k +\sigma^2_\delta)$. This is due to the fact that the vector $\mathbf{v}^{\mathcal{S}_0}_k$, in this case, spans $M-K-L+1$ dimensional space due to $K-1+L$ nulls.


\section{Proof of corollary \ref{cor:dmpvup_P_k}}
\label{app_cor_dmpvup_P_k}
Consider two sets $\mathcal{S}_1$ and $\mathcal{S}_2$ containing SU-$k$ such that $\mathcal{S}_2 \subset \mathcal{S}_1$ and $|\mathcal{S}_1|> |\mathcal{S}_2|$. Since $P^{\mathcal{S}_1}_k$ and $P^{\mathcal{S}_2}_k$ are Gamma random variables, their CDFs can be written as follows:
\begin{align}
\nonumber \Pr(P^{\mathcal{S}_1}_k \leq x) = \frac{1}{\Gamma(\kappa^p_k)} \int_{0}^{x/(\gamma^{\mathcal{S}_1}_k\theta^p_k)} t^{\kappa^p_k-1}e^{-t}dt
\\\Pr(P^{\mathcal{S}_2}_k \leq x) =\frac{1}{\Gamma(\kappa^p_k)} \int_{0}^{x/(\gamma^{\mathcal{S}_2}_k\theta^p_k)} t^{\kappa^p_k-1}e^{-t}dt.
\end{align}
Since $\gamma^{\mathcal{S}_1}_k > \gamma^{\mathcal{S}_2}_k$, we get $\Pr(P^{\mathcal{S}_1}_k \leq x) < \Pr(P^{\mathcal{S}_2}_k \leq x)$ or $\Pr(P^{\mathcal{S}_1}_k \geq x) > \Pr(P^{\mathcal{S}_2}_k \geq x)$. 


\section{Derivation of $\Pr(R^{\mathcal{S}}_k \geq y)$}
\label{app_R_k}
The rate achieved at SU-$k$ is given as
\begin{align}
R^{\mathcal{S}}_{k} = \log_2 \left( 1 + \frac{P^{\mathcal{S}}_k |\mathbf{ h}_{k}^H \mathbf{v}^{\mathcal{S}}_k|^2}{\sigma^2_w + I_k+ \sum \limits_{j \in \mathcal{S}, j \neq k} I_{jk} }\right), k \in \mathcal{S},
\end{align}
where $\mathcal{S}$ is the selected set. Substituting for $P_k^{\mathcal{S}}$ from (\ref{eq:P_k_vup}) in the above equation, we get
\begin{align}
\nonumber \Pr(R^{\mathcal{S}}_k \geq y) &=  \\
&\Pr \left(\frac{|\mathbf{ h}_{k}^H \mathbf{v}^{\mathcal{S}}_k|^2}{|\mathbf{ \hat{h}}_{k}^H \mathbf{v}^{\mathcal{S}}_k|^2} \frac{\sigma^2_w + {I}_k + \epsilon_2}{\sigma^2_w + I_k+ \sum \limits_{j \in \mathcal{S}, j \neq k} I_{jk} } \geq \frac{2^y-1}{2^{R^0_k}-1}\right).
\label{eq:app_Pr_R_K}
\end{align}
Similar to the variable $Y$ in the previous appendix, variables $|\mathbf{ h}_{k}^H \mathbf{v}^{\mathcal{S}}_k|^2$ and $|\mathbf{\hat{ h}}_{k}^H \mathbf{v}^{\mathcal{S}}_k|^2$ are approximated with their average values $(M-|\mathcal{S}|-L+1)\beta_k$ and $(M-|\mathcal{S}|-L+1)(\beta_k + \sigma^2_\delta)$, respectively. Let us define $C_y$ as follows:
\begin{align}
C_y = \frac{|\mathbf{ h}_{k}^H \mathbf{v}^{\mathcal{S}}_k|^2}{|\mathbf{ \hat{h}}_{k}^H \mathbf{v}^{\mathcal{S}}_k|^2}\left(\frac{2^{R^0_k}-1}{2^y-1}\right) =  \frac{\beta_k}{\beta_k + \sigma^2_\delta} \left(\frac{2^{R^0_k}-1}{2^y -1}\right).
\end{align}
Substituting the above equation, ${I}_{k}=   \sum_l P_p |{{h}_{lk}}|^2 $, and $I_{jk}=P_j |\mathbf{  h}_{k}^H \mathbf{v}_j|^2$ in (\ref{eq:app_Pr_R_K}), we can rewrite the equation as follows:
\begin{align}
&\nonumber \Pr(R^{\mathcal{S}}_k \geq y) =\\\nonumber &\Pr\left( (1-C_y)\sum_{l\in \mathcal{T}}P_p|{h_{lk}}|^2+ \sum_{j \in \mathcal{S}, j \neq k}P^{\mathcal{S}}_j |\mathbf{  h}_{k}^H \mathbf{v}_j|^2\leq \zeta_y\right)
\\= &\Pr\left(\sum_{l\in \mathcal{T}}Z_{lk} + Z_k \leq \zeta_y \right),
\label{eq:app_P_R_K_2}
\end{align}
where $Z_{lk}=(1-C_y)P_p|{h_{lk}}|^2, Z_k= \sum_{j \in \mathcal{S}, j \neq k}P^{\mathcal{S}}_j |\mathbf{  h}_{k}^H \mathbf{v}^{\mathcal{S}}_j|^2 $, and $\zeta_y$ as defined in (\ref{eq:kz_thetaz}). Since $Z_{lk}$ and $Z_k$ are independent random variables, the cdf in the RHS of (\ref{eq:app_P_R_K_2}) can be expressed in terms of Fourier transforms of the characteristic functions these variables. Therefore, we derive the characteristic functions of $Z_{lk}$ and $Z_k$. The variable $Z_{lk}$ is a Gamma random variable $\sim \Gamma(1,(1-C_y)P_p\beta_{lk})$ with characteristic function:
\begin{align}
\phi_{lk}(jt) = (1-\theta^z_{lk}jt),
\label{eq:app_phi_lk}
\end{align}
where $\theta^z_{lk} = (1-C_y)P_p\beta_{lk}$.
Further, the variable $Z_k= \sum_{j \in \mathcal{S}, j \neq k}P^{\mathcal{S}}_j |\mathbf{h}_{k}^H \mathbf{v}^{\mathcal{S}}_j|^2$ can be written as $Z_k= \sum \limits_{j \in \mathcal{S}, j \neq k}P^{\mathcal{S}} |\mathbf{\hat{h}}_{k}^H \mathbf{v}^{\mathcal{S}}_j  + \mathbf{\delta}_k^H \mathbf{v}^{\mathcal{S}}_j|^2$ $=  \sum_{j \in \mathcal{S}, j \neq k}P^{\mathcal{S}}_j |\mathbf{\delta}_k^H \mathbf{v}^{\mathcal{S}}_j|^2$. The second equality follows from $\mathbf{\hat{h}_k v^{\mathcal{S}}_j} =0$ due to zero forcing beamforming. The term $|\mathbf{\delta}_k^H \mathbf{v}^{\mathcal{S}}_j|^2$ is the projection of isotropic vector $\delta_k \sim CN(0,\sigma^2_\delta)$ on uncorrelated space spanned by $\mathbf{v}^{\mathcal{S}}_k$, which gives $|\mathbf{\delta}_k^H \mathbf{v}^{\mathcal{S}}_j|^2 \sim \Gamma(1,\sigma^2_\delta)$ \cite[lemma 3]{hosseini2014}. Therefore, $P^{\mathcal{S}}_j |\mathbf{\delta}_k^H \mathbf{v}^{\mathcal{S}}_j|^2$ is a product of two Gamma random variables and is approximated as a Gamma random variable \cite{coelho2014}. The mean and the variance of $P^{\mathcal{S}}_j |\mathbf{\delta}_k^H \mathbf{v}^{\mathcal{S}}_j|^2$ are given below:
\begin{align}
\nonumber\mathbb{E}[P^{\mathcal{S}}_j |\mathbf{\delta}_k^H \mathbf{v}^{\mathcal{S}}_j|^2]&= \sigma^2_\delta \gamma^{\mathcal{S}}_j \theta^p_j \frac{\Gamma(\kappa^p_j+1)}{\Gamma(\kappa^p_j)}
\\ \mathbf{var}[P^{\mathcal{S}}_j |\mathbf{\delta}_k^H \mathbf{v}^{\mathcal{S}}_j|^2]&= (\sigma^2_\delta \gamma^{\mathcal{S}}_j \theta^p_j)^2 \frac{2\Gamma(\kappa^p_j+2)\Gamma(\kappa^p_j)- \Gamma^2(\kappa^p_j+1)}{\Gamma^2(\kappa^p_j)}.
\label{eq:app_mean_var_zk}
\end{align}
The variable $Z_k$ is modeled as a Gamma random variable with shape parameter $\kappa^z_k$ and shape scale parameter $\theta^z_k$. The parameters are computed using moment matching method \cite[lemma 3]{hosseini2014} by solving the following two equations for $\kappa^z_k$ and $\theta^z_k$:
\begin{align}
\nonumber \kappa^z_k{\theta^z_k}=\sum_{j \in \mathcal{S}, j \neq k} \mathbb{E}[P^{\mathcal{S}}_j |\mathbf{\delta}_k^H \mathbf{v}^{\mathcal{S}}_j|^2],
\\\kappa^z_k({\theta^z_k})^2 =\sum_{j \in \mathcal{S}, j \neq k} \mathbf{var}[P^{\mathcal{S}}_j |\mathbf{\delta}_k^H \mathbf{v}^{\mathcal{S}}_j|^2].
\label{eq:app_kz_thetaz}
\end{align}
Expressions in (\ref{eq:kz_thetaz}) follow from (\ref{eq:app_mean_var_zk}) and (\ref{eq:app_kz_thetaz}). The characteristic function of the Gamma random variable $Z_k$ is as follows \cite{Garcia2008}:
\begin{align}
\phi_k(jt) = (1-\theta^z_k jt )^{-\kappa^z_k}.
\label{eq:app_phi_k}
\end{align}
Since $Z_{lk}$ and $Z_k$ are independent random variables, the cdf in the RHS of (\ref{eq:app_P_R_K_2}) can be written in terms of the Fourier transform of the product of characteristic functions of these random variables as follows:
\begin{align}
\Pr(R^{\mathcal{S}}_k \geq y) = \frac{1}{2\pi}\int \limits_{-\infty}^{\zeta_y} \int\limits_{-\infty}^{\infty}\left(\prod_{l\in \mathcal{T}}\phi_{lk}(jt)\right) \phi_k(jt) e^{-j2\pi wt} dt dw.
\end{align}
Substituting for $\phi_{lk}(jt)$ and $\phi_k(jt)$ from (\ref{eq:app_phi_lk}) and (\ref{eq:app_phi_k}), respectively, we obtain the expression (\ref{eq:Pr_R_k}). The expression for $\Pr(R^{\mathcal{S}}_k \geq y)$ under DMP without vector update is computed by following the above derivation and replacing $\gamma^{\mathcal{S}}_j$ by $\gamma^{\mathcal{S}_0}_j$ in (\ref{eq:app_mean_var_zk}).

If we have $\sigma^2_\delta = 0$, $\mathbf{\hat{h}}_{k}=\mathbf{{h}}_{k}$, and $I_{jk}=0$. Then, substituting $y=R^0_k$ in
(\ref{eq:app_Pr_R_K}), we get $\Pr(R^\mathcal{S}_k \geq R^0_k)=1$.

%
%
%


\section{Derivation of $\mathbb{E}[I_{lk}]$}
\label{app_I_kl}
In DMP, we have $\mathbb{E}\left[I_{kl}\right] =\mathbb{E}[P^{\mathcal{S}}_k |\mathbf{\Delta}^H_{l0}\mathbf{v}^{\mathcal{S}}_k|^2]$. The expression $\mathbb{E}[P^{\mathcal{S}}_k |\mathbf{\Delta}^H_{l0}\mathbf{v}^{\mathcal{S}}_k|^2]$ is obtained by following the derivation of $\mathbb{E}[P^{\mathcal{S}}_j |\mathbf{\delta}_k^H \mathbf{v}^{\mathcal{S}}_j|^2]$ in Appendix \ref{app_R_k} and replacing $\mathbf{\delta}_k$ and $\sigma^2_{\delta}$ with $\mathbf{\Delta}_{l0}$ and $\sigma^2_{\Delta}$, respectively. In DMP without vector update, the expression for $\mathbb{E}\left[I_{kl}\right]$ is obtained by replacing $\gamma^{\mathcal{S}}_k$ with $\gamma^{\mathcal{S}_0}_k$ in (\ref{eq:I_kl_1}).


\bibliographystyle{IEEEtran}
\bibliography{IEEEabrv,final_refs}

\end{document}